\documentclass[11pt]{article}

\usepackage{typearea}
\typearea{14}
\usepackage{setspace}
\usepackage[breaklinks]{hyperref}
\hypersetup{colorlinks=true,%
            citebordercolor={.6 .6 .6},linkbordercolor={.6 .6 .6},%
citecolor=blue,urlcolor=black,linkcolor=blue,pagecolor=black}
\usepackage[compact]{titlesec}

\usepackage{amsmath,amsthm,amssymb}
\usepackage{amsthm,nicefrac}
\usepackage{mdframed}
\usepackage{thmtools,thm-restate}
\usepackage{graphicx}
\usepackage{subcaption}
\usepackage{amsfonts,mathtools}
\usepackage[ruled,vlined,linesnumbered,algonl]{algorithm2e}
\SetEndCharOfAlgoLine{}

\usepackage{soul}
\usepackage[dvipsnames]{xcolor}
\definecolor{fluorescentyellow}{rgb}{0.8, 1.0, 0.0}
\sethlcolor{fluorescentyellow}

\usepackage{xspace}
\usepackage[nameinlink]{cleveref}
\usepackage[font={small,it}]{caption}

\makeatletter
\allowdisplaybreaks
\makeatother

\DeclareMathOperator{\poly}{poly}
\DeclareMathOperator{\cost}{cost}
\DeclareMathOperator{\demand}{demand}
\DeclareMathOperator{\sparsity}{sparsity}
\DeclareMathOperator{\dist}{dist}
\DeclareMathOperator{\diameter}{diameter}

\newcommand{\RR}{\mathbb{R}}
\newcommand{\ZZ}{\mathbb{Z}}
\newcommand{\E}{\mathbb E}
\newcommand{\partition}{\pi}
\newcommand{\set}[1]{\{#1\}}

\newcommand{\eps}{\varepsilon}
\newcommand{\e}{\eps}
\newcommand{\nf}{\nicefrac}
\newcommand{\Y}{Z}
\newcommand{\sse}{\subseteq}
\newcommand{\amenable}{amenable\xspace}
\newcommand{\friendly}{friendly\xspace}

\newcommand{\calA}{\mathcal{A}}
\newcommand{\calC}{\mathcal{C}}
\newcommand{\calE}{\mathcal{E}}
\newcommand{\calP}{\mathcal{P}}
\newcommand{\calT}{\mathcal{T}}
\newcommand{\calV}{\mathcal{V}}

\newcommand{\bT}{\mathcal{T}}
\newcommand{\bbF}{\mathbb{F}}

\newcommand{\barp}{\kappa}

\newtheorem{theorem}{Theorem}[section]
\newtheorem{definition}[theorem]{Definition}
\newtheorem{proposition}[theorem]{Proposition}

\newtheorem{lemma}[theorem]{Lemma}
\newtheorem{fact}[theorem]{Fact}
\newtheorem{observation}[theorem]{Observation}
\newtheorem{claim}[theorem]{Claim}
\newtheorem{subclaim}[theorem]{Subclaim}

\newtheorem{corollary}[theorem]{Corollary}

\newtheorem{assumption}[theorem]{Assumption}

\newcommand{\ignore}[1]{}

\newcounter{note}[section]

\newcommand{\initOneLiners}{%
    \setlength{\itemsep}{0pt}
    \setlength{\parsep }{0pt}
    \setlength{\topsep }{0pt}
      \usecounter{myLISTctr}
}
\newenvironment{OneLiners}[1][\ensuremath{\bullet}]
{\begin{list}
    {#1}
    {\initOneLiners}}
  {\end{list}}

  \title{A Quasipolynomial $(2+\eps)$-Approximation for Planar Sparsest Cut}

  \author{Vincent Cohen-Addad\thanks{Google Research, Zurich,
      Paris.} \and 
    Anupam Gupta\thanks{Carnegie Mellon University, Pittsburgh PA
      15217.} \and
    Philip N. Klein\thanks{Brown University.  Research
      supported by NSF Grant CCF-1841954.} \and
    Jason Li$^\dagger$
  }
  \date{}

\newif\ifverbosity
\verbosityfalse

\newenvironment{explanation}{}{}

\begin{document}

\maketitle

\begin{abstract}
  The (non-uniform) sparsest cut problem is the following %
  graph-partitioning problem: given a ``supply'' graph, and demands on
  pairs of vertices, delete some subset of supply edges to minimize
  the ratio of the supply edges cut to the total demand of the pairs
  separated by this deletion. %
  Despite much effort,
  there are only a handful of nontrivial classes of supply graphs for
  which constant-factor approximations are known.

  We consider the problem for planar graphs, and give a
  $(2+\eps)$-approximation algorithm that runs in quasipolynomial
  time.  Our approach defines a new structural decomposition of an
  optimal solution using a ``patching'' primitive.
  We combine this
  decomposition with a Sherali-Adams-style linear programming
  relaxation of the problem, which we then round. This should be 
  compared with the polynomial-time approximation algorithm of Rao~(1999), which
  uses the metric linear programming relaxation and $\ell_1$-embeddings, 
    and achieves an $O(\sqrt{\log n})$-approximation in polynomial time.
\end{abstract}

\section{Introduction}
\label{sec:introduction}

In the \textit{(non-uniform) sparsest cut} problem, we are given a
``supply'' graph $G$ with an edge-cost function
$\cost: E(G) \to \RR_{\geq 0}$ and a \emph{demand} function
$\demand: V \times V \to \RR_{\geq 0}$.  For a nonempty proper subset $U$ of the
vertices of $G$, the corresponding \emph{cut} is the edge subset
$\delta_{G}(U)=\{\set{x,y} \in E(G) \mid x\in U, y\not\in U\}$. The cost of
this cut is $\cost(U) := \sum_{e\in \delta_{G}(U)} \cost(e)$, and the
\emph{demand separated} by it is
$\demand(U) := \sum_{\set{u,v}: |\set{u,v} \cap U| =1} \demand(u,v)$. The
\emph{sparsity} of the cut given by $U$ is $\Phi(U) := \cost(U)/\demand(U)$.  The goal
is to find a set $U$ that achieves the \emph{minimum sparsity} of this
instance, defined as:
\begin{gather}
  \Phi^* %
  := \min_{U: U \neq \emptyset, V(G)} \Phi(U) = \min_{U: U \neq \emptyset, V(G)} \frac{\cost(U)}{\demand(U)} . \label{eq:def-sc} %
\end{gather}
(The
special case with unit demand between every pair of vertices is called
the \emph{uniform} sparsest cut, discussed in
\S\ref{sec:related-work}.) Finding sparse cuts is a natural clustering
and graph decomposition subroutine used by divide-and-conquer
algorithms for graph problems, and hence has been widely studied.

The problem is NP-hard~\cite{MS90}, and
so %
the %
focus has been on the design of approximation
algorithms. This line of work started with an $O(\log D \log C)$-approximation
given by Agrawal, Klein, Rao, and Ravi~\cite{KleinRAR95,KleinARR90},
where $D$ is the sum of demands and $C$ the sum of capacities.
After a several developments, the
best approximation factor currently known is 
$O(\sqrt{\log n} \log \log n)$ due to Arora, Lee, and
Naor~\cite{ALN05}. Moreover, the problem is inapproximable
to any constant factor, assuming the unique games
conjecture~\cite{CKKRS05,KV05}.

Given this significant roadblock for general graphs, a major research effort has sought
$O(1)$-approx\-imation algorithms for ``interesting'' classes of graphs. In
particular, the problem restricted to the case where $G$ is
\emph{planar} has received much attention over the years.
The
best approximation bound for this special case before the current work
remains the $O(\sqrt{\log n})$-approximation of Rao~\cite{Rao99},
whereas the best hardness result merely rules out a
$\frac{1}{0.878 + \eps} \approx (1.139 - \eps)$-approximation assuming
the unique games conjecture~\cite{GTW13-spcut}.  One source of the
difficulty is that the ``demand'' graph, i.e., the support of the
demand function, is not necessarily planar.\footnote{If the union of
  this graph with the input graph (the supply graph) is planar then
  the problem admits better approximation algorithms.}
  Indeed, the hardness results are
obtained by embedding general instances of max-cut in this demand graph.

The sparsest cut problem has been studied on even more specialized
classes of graphs in order to gain insights for the planar case: see,
e.g., \cite{OS81,GNRS99,CGNRS01,CJLV08,CSW10,LR10,CFW12}. Again, despite
successes on those specialized classes (see~\S\ref{sec:related-work}
for a discussion), getting a constant-factor approximation for
non-uniform sparsest-cut on arbitrary planar graphs has remained open so 
far. Our main result is such an algorithm, %
at the expense of quasi-polynomial runtime:
\begin{theorem}[Main Theorem]
  \label{thm:main}
  There is a randomized algorithm for the non-uniform sparsest cut
  problem on planar graphs which achieves a $(2+\eps)$-approximation
  in quasi-polynomial time $O(n^{\log^2 n/\poly(\eps)})$.
\end{theorem}

\subsection{Related Work}
\label{sec:related-work}

The sparsest cut problem is NP-hard, even for the uniform
case~\cite{MS90}. The initial approximation factor of
$O(\log D \log C)$ \cite{KleinRAR95,KleinARR90}, where $D$ is the sum of demands and
$C$ the sum of capacities, was
improved to $O(\log^2 n)$ \cite{PT95}, %
to $O(\log n)$ \cite{LLR95,AR98}, $O(\log^{3/4} n)$ \cite{CGR05}, and
finally to the current best
$O(\sqrt{\log n} \log \log n)$ \cite{ALN05}.

Many special classes of graphs admit constant-factor approximations
for non-uniform sparsest cut. These include outerplanar~\cite{OS81},
series-parallel~\cite{GNRS99,CSW10,LR10}, $k$-outerplanar
graphs \cite{CGNRS01}, graphs obtained by $2$-sums of
$K_4$ \cite{CJLV08}, graphs with constant pathwidth and  related
families \cite{LS09,LeeS13}, and bounded treewidth graphs \cite{CKR10,
  GTW13-spcut}.  %
Most of these approximations are with respect to the ``metric''
relaxation LP, via $\ell_1$-embeddings~\cite{LLR95}. The bounded
treewidth results are exceptions: they use a stronger LP, and it
remains open whether the metric relaxation has a constant integrality
gap even for graphs of treewidth $3$, which are also planar. Neither
our result, nor the results of~\cite{CKR10, GTW13-spcut}, shed light
on this question.  A ``meta-result'' of Lee and
Sidiropoulos~\cite{LS09} proves that if planar graphs embed into
$\ell_1$ with constant distortion, and if constant-distortion
$\ell_1$-embeddability is closed under clique-sums, then all
minor-closed families have constant-factor approximations with respect
to the standard LP.

\emph{Uniform Version.}  For the \emph{uniform} version of sparsest
cut, $O(1)$-approximations exist for classes of graphs that exclude
non-trivial minors, via low-diameter decompositions~\cite{KPR93}.
Park and Philips~\cite{Park1993FindingMC} solve the uniform problem on
planar graphs in $O(n^3)$ time; see Abboud et
al.~\cite{AbboudSparsest} for a speedup.  Abboud et al.~\cite{AbboudSparsest}
give an
$O(1)$-approximation in near-linear time, improving upon the
near-linear time $O(\log n)$-approximation of Rao~\cite{Rao92}.
Finally, Patel~\cite{Patel13} showed that the uniform sparsest-cut
problem can be exactly solved in time $n^{O(g)}$ for graphs that embed
into a surface of genus at most $g$.  The approximation factor for the
uniform problem on general graphs has been improved from
$O(\log n)$~\cite{LR10} by rounding LP relaxations, to
$O(\sqrt{\log n})$~\cite{ARV04} by rounding SDP relaxations.

\emph{Integrality Gaps for the Basic LP:} There are many results on
lower bounds and limitations to known techniques; see, e.g., works by
Khot and Vishnoi~\cite{KV05}, Chawla et al.~\cite{CKKRS05}, Lee and
Naor~\cite{LeeNaor06}, and others. Naor and
Young~\cite{NaorY17,NaorY18} show lower bounds on integrality gap of
the semidefinite program for sparsest-cut that almost match the upper
bounds of~\cite{ALN05}.

\emph{Hardness of Approximation:} an $\alpha$-approximation to
non-uniform sparsest cut on series-parallel (i.e., treewidth-$2$)
graphs, which are also planar, gives an $(1/\alpha)$-approximation to
the \textsc{MaxCut} problem~\cite{GTW13-spcut}. The results
of~\cite{Hastad01,KKMO} now imply that sparsest cut is hard to
approximate better than $\frac{17}{16}$ unless P=NP, and to better
than $1/0.878$ assuming the Unique Games conjecture. It is known that
for every $\e>0$, there are graph families with treewidth $f(\e)$ on
which sparsest-cut is hard to approximate better than a factor of
$(2-\eps)$, assuming the Unique games conjecture---but these hard
instances are not known to be planar~\cite{GTW13-spcut}.

\subsection{Techniques}
\label{sec:techniques}

We outline the main conceptual steps of our algorithm, and some
intuition for why these are needed:

\begin{enumerate}
\item \emph{Duality.} One advantage of planar instances is the duality between cuts
  and cycles. Indeed, there is an optimal solution that corresponds to
  a simple cycle $C_0$ in the dual graph $G^*$. So it suffices to find
  some cycle in the dual with low total edge-cost, that separates a lot
  of demand (which is now between pairs of faces), see \Cref{lem:simple-cut}.

\item \emph{Low-Complexity Clusterings.} Suppose we efficiently find a hierarchical partition of
  the dual into subgraphs called \emph{clusters}, such that whenever some cluster $K$ splits into
  subclusters $K_1, K_2, \ldots K_t$, the cycle $C_0$ crosses between
  these subclusters at most $O(\log n)$ times. (I.e., $C_0$ has ``low
  complexity'' with respect to the partition.) Then, we can
  find a $2$-approximation to the sparsest ``low-complexity''
  solution using a linear program, as described in \Cref{item:lp} below.
  
\item \emph{Finding these Low-Complexity Clusterings.} \label{item:low-comp} How do we find such a good hierarchical decomposition? We
  repeatedly find low-diameter decompositions with decreasing radii.
  If the cost of the edges of $C_0$ that lies within each cluster
  $K$ is at most $O(\eps^{-1} \log n)$ times the diameter of $K$, then
  performing a low-diameter decomposition of $K$ causes only
  $O(\eps^{-1} \log n)$ edges to be cut in expectation. I.e., the expected number of
  times $C_0$ crosses between subclusters of $K$ is small, as
  desired. (Observe we get a small number of crossings only in
  expectation: we'll come back to this in \Cref{item:nondet}).

\item \emph{Patching.} However, the cost of $C_0$ within some cluster $K$ may exceed
  $\Theta(\eps^{-1} \log n)$ times the cluster diameter. In this case
  we \emph{patch} the cycle, adding some collection of shortest paths
  within the cluster $K$, that locally break this cycle into several
  smaller cycles. We elaborate on this operation in the paragraphs
  following this outline.  If we imagine maintaining a collection of
  cycles, starting with the single cycle $C_0$, the patching replaces
  one cycle in this collection by many. Moreover the cost of this
  collection increases by a factor of at most $(1 + \eps/\log n)$ in each
  level of the recursion. 

\item \emph{Controlling the Cost.} Since the ratio of the largest
  to the smallest edge cost can be assumed to be polynomially bounded (see Section~\ref{sec:preliminaries})
  and the diameter decreases by a factor of $2$ at each level of the decomposition,
  there are $O(\log n)$ levels of recursion. This means the total
  increase in the cost of the entire collection of cycles is at most
  $(1+\eps/\log n)^{O(\log n)} = (1+ O(\eps))$.  So
  the sparsest simple cycle from this collection (and its
  corresponding simple cut in the primal) has 
  sparsity at most $(1+ O(\eps))$ times the optimal sparsity. %

\item \emph{``Non-Deterministic'' Hierarchical
    Decompositions.}\label{item:nondet} Recall, in
  \Cref{item:low-comp} we ensured the low-complexity property only in
  expectation. We need it to hold for all the clusters of the
  decomposition, and so with high probability for a single cluster. To
  achieve this we choose
  $\Theta(\log n)$ independent low-diameter decompositions for each
  cluster, and apply the procedure recursively to each part of each
  partition. This is reminiscent of an idea
  introduced by Bartal et al.~\cite{BartalGK16}. It ensures that one of these partitions has low
  complexity, with high probability. We call this a \emph{non-deterministic hierarchical
    decomposition}, and show that it has total size
  $n^{O(\eps^{-1}\log n)}$.

\item \emph{The Linear Program.}\label{item:lp} All the above steps
  were part of the structure lemma. They show the existence (whp.) of
  a near-optimal low-complexity solution with respect to the
  non-deterministic hierarchical decomposition. It now remains to
  select one of the decompositions at each level, and to find this
  cycle that has low complexity with respect to this restricted
  decomposition tree. To do this, we write a linear program, and round
  it. The high level ideas are similar to those used for the
  sparsest-cut problem on bounded-treewidth graphs, and we elaborate
  on these in the paragraph below.
\end{enumerate}

Before we proceed, a caveat: the actual algorithm differs from the above outline
in small but important details; e.g., we coarsen the low-diameter
decompositions to ensure that each partition has few parts, which
means the diameter of our clusters does not necessarily drop
geometrically. Since these details would complicate the above
description, we introduce them only when needed.

We now give more details about two of the key pieces: the patching
lemma, and the linear program.

\subsubsection{A \emph{Patching Lemma} for Planar Graphs}
An important ingredient of our approach is a \emph{patching lemma} for
non-uniform sparsest cut in planar graphs. As with most patching lemmas,
our patching lemma (and the associated
patching procedure) are only used for the analysis of our algorithm;
Their goal is to help exhibit a near-optimum solution that is ``well-structured''.
There are some similarities
to the patching lemmas of Arora~\cite{Arora97} and Bartal et
al.~\cite{BartalGK16} for the Traveling Salesman Problem (TSP) in
Euclidean space and doubling metrics respectively, but there are some
important differences. Given a cluster $K$ of diameter $D$ and a given cycle $C$
(thought of as the optimum solution), the goal of our patching
procedure is to break $C$ into a collection of cycles such that
(1)~for each demand pair separated by $C$, there is a cycle in the
collection separating that pair; (2)~for each cycle $C'$ in the
collection, the total cost of the edges of $C'$ inside cluster $c$ is
$O(\eps^{-1} \log n)$; and (3) the total cost of the edges of the
cycles of the collection inside cluster $c$ is at most
$(1+\eps/\log n)$ times the total cost of the edges of cycle $C$
inside $K$. %
We explain our patching procedure
in \Cref{fig:techniques}; please read the captions to follow along.

\begin{figure}[h]
  \begin{subfigure}{.46\textwidth}
    \centering
    \includegraphics[width=2in,page=1]{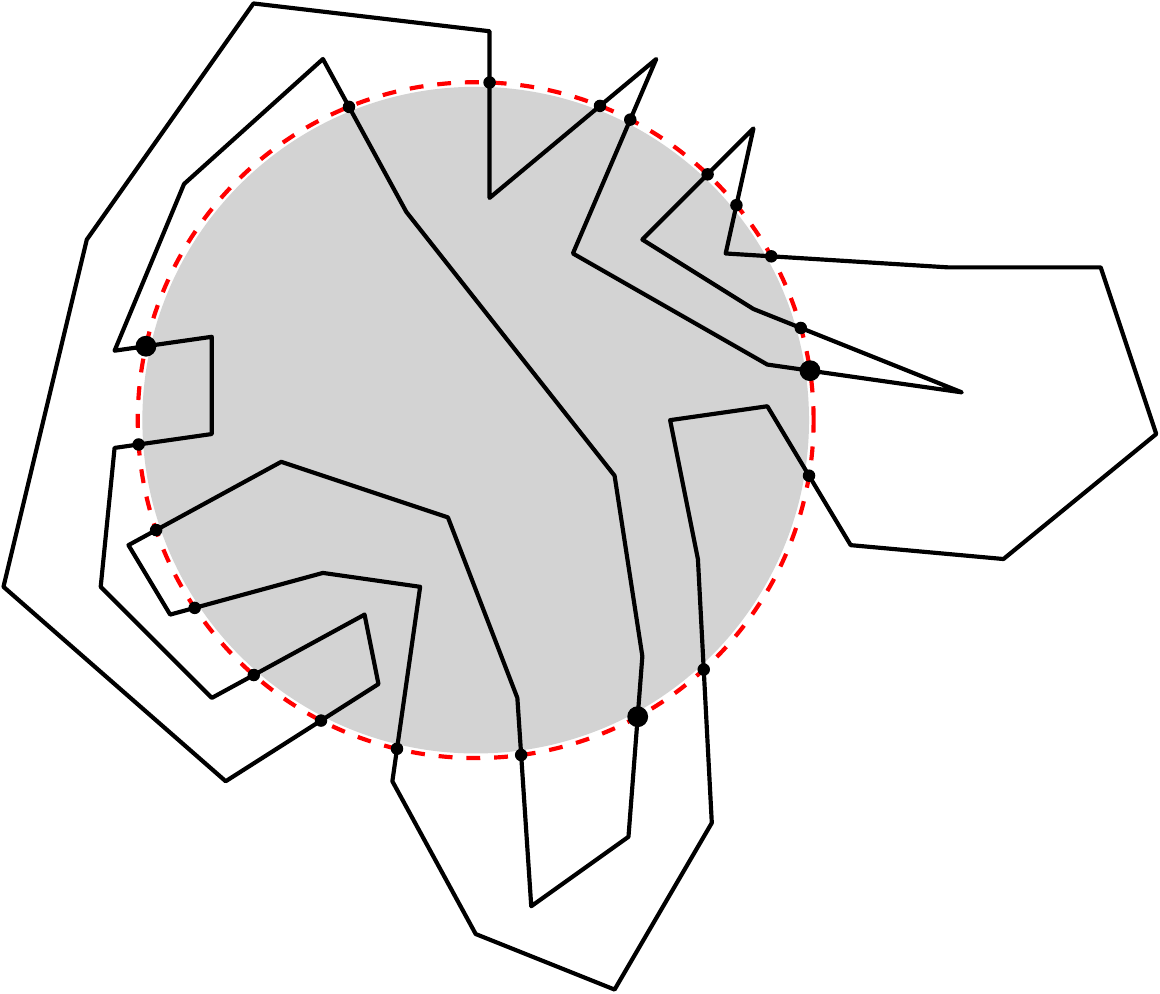}
    \caption{The original cycle $C$ in black, and the cluster $K$ in
      dashed red. The three black dots are vertices on the optimal
      solution such that the edge-cost
      of the optimal solution $C$ within $K$
      between any two consecutive dots is $\Theta(\frac{\log n}{\eps} D)$,
      where $D$ is the diameter of $K$.}
  \end{subfigure}%
  \hfill
  \begin{subfigure}{.46\textwidth}
  \centering
    \includegraphics[width=2in,page=2]{figures/algo.pdf}
    \caption{Since the cost of $C$ within $K$ between consecutive
      black dots is $\Theta(\eps^{-1}\log n)$ times the diameter of
      $K$, adding the shortest paths (shown in blue) from some root
      node $r$ to the dots on $C$ only increases the cost by a factor
      of $(1+\nicefrac{\eps}{\log n})$.}
\end{subfigure}\\
  \bigskip
  
\centering
\begin{subfigure}{.32\textwidth}
  \centering~~~~~
  \includegraphics[width=1.5in,page=4]{figures/algo.pdf}
\end{subfigure}
\begin{subfigure}{.32\textwidth}
  \centering~~~~~
  \includegraphics[width=1.5in,page=5]{figures/algo.pdf}
\end{subfigure}
\begin{subfigure}{.32\textwidth}
  \centering~~~~~
  \includegraphics[width=1.5in,page=3]{figures/algo.pdf}
\end{subfigure}
\caption{The addition of these paths splits the original tour into
  three tours (possibly self-intersecting), shown in purple, orange,
  and red. One of these three tours achieves sparsity at most that of
  the original tour (up to a $(1+\nicefrac{\eps}{\log n})$ factor).
  Importantly, the cost within $K$ of this good tour is
  $O(\frac{\log n}{\eps} D)$, by construction, which is %
  proportional to the diameter of the
  cluster. Repeating this argument recursively, and using that we have
  $O(\log n)$ levels, the total error is bounded by a multiplicative
  factor of $(1+\eps)$.  }
\label{fig:techniques}
\end{figure}

Obtaining a patching lemma for a planar problem seems surprising to
us. While low-diameter decompositions have been widely used to obtain
approximation schemes for problems in Euclidean spaces of constant
dimension
(e.g., for the traveling salesman problem, or facility location), it
was unclear how to use low-diameter decompositions effectively to obtain
approximation schemes for these problems on planar graphs. One hurdle
in applying this technique to planar graphs has been that that the isoperimetric
inequality does not hold, and so the cost of the edges of the cluster boundary cannot be
related to the diameter of the cluster. Without isoperimetry, there
are examples where forcing the optimum solution to
make a small number of crossings between child clusters (through
\emph{portals}, for example) incurs a huge increase in cost. %
One could get a coarser control on the problem structure,
e.g., by bounding the diameter, which gave constant-factor approximations for
related problems such as multicut or $0$-extension,
but only an $O(\sqrt{\log n})$-approximation for
sparsest cut. %
The approach of enriching the optimum
circumvents this issue for the non-uniform sparsest cut
problem.

\subsubsection{A Linear Program over these Clusterings}
While the natural candidate for finding the best solution over a
non-deterministic hierarchical decomposition would be a dynamic
program, we currently do not know how to get a good approximation
for sparsest cut using dynamic programming, even for %
simpler graph classes such as treewidth-$2$ graphs.
However, we can use linear programs as in~\cite{CKR10,GTW13-spcut}:
we
add linear constraints capturing that the LP has to ``choose'' one of
the $\Theta(\log n)$ potential sub-partitionings at each level of the
decomposition. Our linear program has variables that capture all the
``low-complexity'' partial solutions within each cluster, and
constraints that ensure the consistency of the solution and of the
partitioning over all $O(\log n)$ levels of the entire decomposition.
The details appear in \S\ref{sec:sa}.

\section{Notation and Preliminaries}
\label{sec:preliminaries}

Let $[n]$ denote $\set{1,\ldots, n}$. Given an instance of non-uniform
sparsest-cut, 
the following lemma allows us to restrict edge costs and demands on
the pairs to be integers in a bounded range.
It is proved in \S\ref{app:prelims}.

\begin{lemma}
  \label{lem:aspectratio}
  An $\alpha$-approximation algorithm $\calA$ with runtime $T_n$ for
  sparsest cut instances on $n$-vertex planar graphs with the edge
  costs in $[n^2]$ and demands for pairs in %
  $[n^3]$ implies an
  $(1+o(1))\alpha$-approximation algorithm with running time
  $O(n^3 \cdot T_n)$ for planar sparsest cut instances with arbitrary
  non-negative edge costs and demands.
\end{lemma}

For a connected graph $G$ and subset $U \sse V(G)$ of vertices, let
$\delta_G(U)$ denote the set of edges $e$ of $G$ such that $U$
contains exactly one endpoint of $e$.  A set of edges of this form is
called a \emph{cut}.  We drop the subscript when the
graph $G$ is unambiguous.
Let $G[U]$ denote the subgraph induced by $U$. A nonempty
proper cut $\delta_G(U)$ is called a \emph{bond} or a \emph{simple
  cut} if both induced subgraphs $G[U]$ and $G[V(G) \setminus U]$ are
connected. The following simple result from, e.g., \cite{OS81} is
proved
in~\S\ref{app:prelims}.

\begin{proposition}
  \label{lem:simple-cut}
  The optimal sparsity
  in~\eqref{eq:def-sc} can be achieved by a set $U$ such that
  $\delta_G(U)$ is a simple cut.
\end{proposition}

\paragraph{Planar duality}
We work with a connected \emph{planar embedded graph}, a graph $G=(V,E)$
together with an embedding of $G$ in the plane. 
There is a corresponding
planar embedded graph $G^*$, the \emph{planar dual} of $G$.  The vertices of
the dual are the faces of $G$, and the faces of the dual are the
vertices of $G$.  The edges in the dual correspond one-to-one
with the edges in $G$, and we identify each edge in the dual $G^*$ with
the corresponding edge in $G$, and hence having the same cost. The demand function can be
considered as mapping pairs of faces in $G^*$ to non-negative integers.

\begin{proposition}[Simple Cycle (e.g., \cite{planarity}, Ch.~5)]
  \label{lem:simple-cycle} Let $G$ be
  a connected planar embedded
graph, and let $G^*$ be its dual. A set of edges forms a simple cut
$\delta_G(U)$ in $G$ if and only if the edges form a simple cycle $C$
in $G^*$.
\end{proposition}
Fix a face $f_\infty$ of $G^*$, i.e., a vertex of $G$. Think of it as
the infinite face of $G^*$.  Let $C$ be a simple cycle in $G^*$.  By
Proposition~\ref{lem:simple-cycle}, the edges of $C$ are the edges of
some simple cut $\delta_G(U)$ where $U$ is a set of vertices of $G$
such that $f_\infty$ is not in $U$.  
We define the set of faces \emph{enclosed} by $C$ to be $U$, and we
denote this set of faces by $U(C)$.

In view of \Cref{lem:simple-cut,lem:simple-cycle},
seeking a sparsest cut in a planar embedded graph is equivalent to
seeking a simple cycle $C$ in $G^*$ with the objective 
\begin{equation} \label{eq:cycle-objective}
\min \frac{\text{cost of edges of } C}{\sum\set{\demand(\set{f_1,f_2}) \mid f_1
    \text{ enclosed by } C, f_2 \text{ not enclosed by }C}}
\end{equation}
Indeed, the value of the objective~(\ref{eq:cycle-objective}) is the
sparsity of $U(C)$ as defined in the introduction.

\begin{restatable}{lemma}{PreserveDemand}
  \label{lem:preserve-demand-separation}
  Suppose $C$ is a simple cycle in $G^*$. Let $C_1, \ldots, C_k$ be
  cycles such that each edge occurs an odd number of times in $C_1,
  \ldots, C_k$ iff it appears in $C$.  Then, for
  each pair $x,y$ of vertices of $G$, if $x$ and $y$ are separated by
  $C$ then they are separated by some $C_i$.
\end{restatable}

\begin{restatable}{lemma}{ChooseCycle}
  \label{lem:good-cycle}
  Suppose $C_0$ is a simple
  cycle in $G^*$ such that $U(C_0)$ has sparsity $s$.  Let $C_1, \ldots, C_k$
  be simple cycles such that every demand separated by $C_0$ is
  separated by at least one of $C_1, \ldots, C_k$.  Suppose the total cost of edges in $C_1, \ldots,
  C_k$ (counting multiplicity) is at most $1+\eps$ times the
  cost of edges in $C_0$.  Then there is some cycle
  $C_i$ such that $U(C_i)$ has sparsity at most $(1+\eps) s$.
\end{restatable}

\paragraph{Low-Diameter Decompositions and $D$-Bounded Partitions.}
A \emph{low-diameter decomposition scheme} takes a graph and
randomly breaks it
into components of ``bounded'' (strong) diameter, so that
the probability of any edge being cut is small. Concretely,
let $H$ be a graph
with edge costs $\cost(e) \geq 0$. Let $\dist_H(x,y)$ be the
shortest-path distances according to these costs. The \emph{(strong)
  diameter} of a subset $U \sse V(H)$ is the maximum distance between
any two nodes in $U$, measured according to 
$\dist_{H[U]}(\cdot, \cdot)$, the distances within this induced subgraph $H[U]$.  The
\emph{strong diameter} of a partition
$\mathcal P = \{V_1, V_2, \ldots, V_s\}$ of the vertex set $V(H)$ is
the maximum strong diameter of any of its parts. A partition is
\emph{$D$-bounded} if it has strong diameter at most $D$. The next
claim follows from
\cite[Theorem~3]{AbrahamGMW10},\cite[Theorem~4]{AGGNT14-sicomp}:
\begin{theorem}
  \label{thm:padded-dec}
  There exists a constant $\beta > 0$ such that given any undirected
  weighted planar graph $H$ and parameter $D > 0$, there exists a
  distribution $\Pi$ over $D$-bounded vertex partitions such that
  \[ \Pr[ \text{$u,v$ fall in different components} ] \leq \beta
    \;\frac{\cost(\set{u,v})}{D} \]
  for any edge $uv \in E(H)$.
  Moreover, this distribution is sampleable in polynomial time.
\end{theorem}

\newcommand{\NDHC}{NDHC\xspace}

\section{Nondeterministic Clustering}
\label{sec:beginning-algorithm}

Recall that the input to our algorithm is:
\begin{OneLiners}
\item a connected simple planar embedded graph $G = (V,E)$ on $n$ vertices and $O(n)$ edges, 
\item assignments of edge costs $\cost: E \to [n^2]$, and
  demands $D: {V \choose 2} \to [n^3]$ to vertex pairs, and
\item a parameter $\eps \in [0,1]$.
\end{OneLiners}
We want the algorithm to output a cut with sparsity $2(1+O(\eps))$ times the optimal
sparsity. By \Cref{lem:simple-cycle} we can assume that the optimal
cut corresponds to a simple cycle in the planar
dual $G^*$. In what follows, we focus on the dual graph $G^*$; the
algorithm will select such a
simple cycle.  We interpret $\cost(\cdot)$ as an assignment of costs
to the edges of $G^*$.

The algorithm consists of two phases. Ideally, the first phase would
construct a hierarchical clustering of $G^*$ such that there exists a
near-optimal solution $C^*$ that crosses each cluster at most $Z$
times, where $Z$ is polylogarithmic. Given the existence of this
``low-complexity'' solution, the second phase would then use this
clustering to compute some near-optimal solution. We proceed slightly
differently: Our method's first phase instead constructs a
\emph{nondeterministic hierarchical clustering (\NDHC)} of $G^*$, which
represents a large family of hierarchical clusterings.  Fortunately,
our second phase can be adapted to work with this family.

Specifically, our algorithm applies a randomized partitioning scheme to
each cluster in each level of hierarchy to form subclusters. However,
this partitioning scheme just ensures that the ``low-complexity''
property for each partition \emph{in expectation}. This is not good
enough because there are many clusters being partitioned, so some of
the partititions might not be good. To handle this shortcoming, the
procedure produces $t \approx n^{O(\log n)}$ partitions of each cluster
into subclusters, so that one of these partitions is good with high probability.
Since we do not know the target solution $C^*$, we cannot choose among
the $t$ partitions. Instead we produce a representation of all the choices
using an \NDHC. A similar idea of repeating the clustering was
previously used by Bartal, Gottlieb, and Krauthgamer~\cite{BartalGK16}
and in subsequent papers in the context of the TSP on doubling metrics.

\subsection{Definitions} \label{sec:clustering}

\begin{definition}[Nondeterministic Hierarchical Clustering]
  \label{def:NDHA}
  For $t \in \ZZ_+$, a $t$-\emph{nondeterministic hierarchical
    clustering ($t$-\NDHC)} for $G^*$ is a rooted tree $\bT$ with alternating
  levels of \emph{cluster nodes} and \emph{partition nodes}:
  \begin{itemize}
  \item A \emph{cluster node} $c$ corresponds to a set $K(c)$ of
    vertices of $G^*$, called a \emph{cluster}.  Each nonleaf cluster
    node has at most $t$ partition nodes as its children.
  \item A \emph{partition node} $p$ with parent $c$ corresponds to
    a partition $\partition(p)$ of the vertex set $K(c)$.  The node
    $p$ has a child $c'$ for each part $P\in \partition(p)$, where
    $c'$ is a cluster node with $K(c')= P$.
\item
  The root of $\bT$ is a partition node which corresponds to the
  trivial partition where all vertices of the graph are in a single part;
  hence it has a single cluster node as its child.
  \item Each leaf is a cluster node with a singleton cluster.
  \end{itemize}

\end{definition}

\begin{figure}
    \centering
    \includegraphics[width=2.5in]{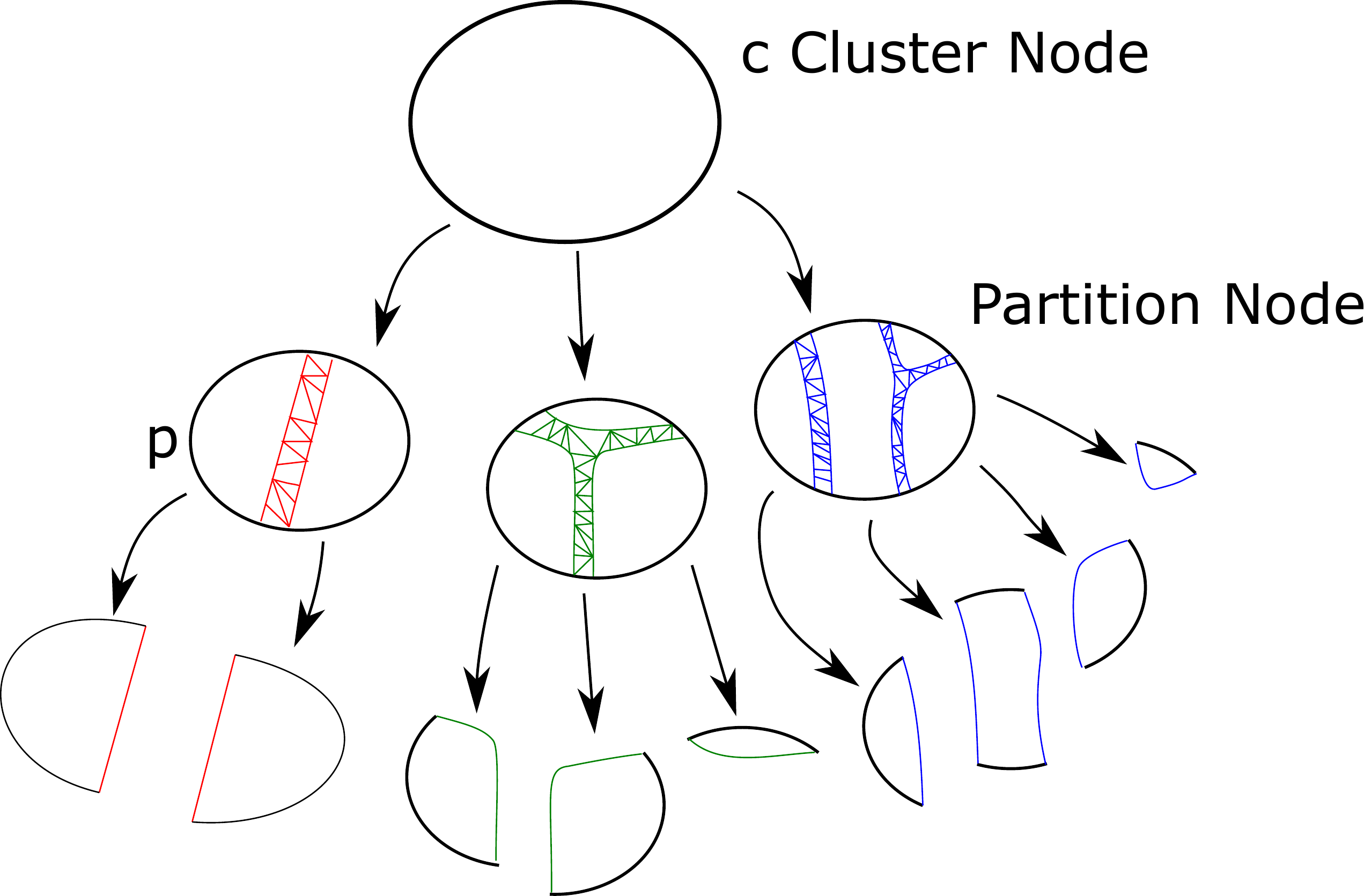}
    \caption{A fragment of tree $\calT$ showing cluster and partition nodes.}
  \label{fig:tree}
\end{figure}

Contrast this structure with a hierarchical partition used in
the literature, which is usually represented by a tree where each
node represents both a cluster (a subset $K(c)$ of vertices of $G^*$),
and a partition of this cluster $K(c)$ into clusters
represented by its children. (In that sense, each node in such a tree
is both a cluster node and a partition node.) We consider several
independent partitions of the same cluster, so we tease these two
roles apart. The usual definition of  hierarchical partition corresponds
to the case where $t=1$. We refer to a $1$-\NDHC as an \emph{ordinary} hierarchical
clustering.

\begin{definition}[Normal and Shattering Partition Nodes]
  \label{def:trivial-and-normal}
  A partition node is called \emph{shattering} if its children cluster
  nodes are all
  leaves, and is otherwise called \emph{normal}.
\end{definition}

\begin{definition}[Part Arity]
  \label{def:part-arity}
  The \emph{part arity} of $\bT$ is the maximum number of children of
  any normal partition node $p$ (which is the same as the maximum
  number of parts in any of the partitions corresponding to these nodes).  We do not limit the number of
  children of a shattering partition node.
\end{definition}

\begin{definition}[Forcing and Relevance]
  \label{def:forcing}
  Given a $t$-NDHC $\mathcal T$
  with $t>1$, define a \emph{forcing} of $\mathcal T$ to be a
  partial function $\varphi$ from cluster nodes to partition nodes
  such that, for each nonleaf cluster node $c$,
  (1) if $\varphi(c)$ is defined then $\varphi(c)$ is a child of $c$
  and (2) if every partition
  node ancestor of $c$ is in the image of $\varphi$ then $\varphi(c)$
  is defined.

 We denote by $\mathcal T\vert_{\varphi}$ the ordinary hierarchical
 decomposition obtaining by retaining only partition nodes $p$ that
 are in the image of $\varphi$ (and also the partition node that is
 the root of $\mathcal T$). 

\end{definition}

\begin{definition}[Internal and Crossing Edges]
  An edge of $G^*$ is \emph{internal} to a set $K$ of vertices if both
  endpoints belong to $K$.
  For a partition $\partition$ of a subset  of the vertices of
  $G^*$, we say an edge \emph{crosses} $\partition$ if the two
  endpoints of the edge lie in two different parts of $\partition$.
  (This requires that the edge be internal to the subset.)
  We use $\delta_{G^*}(\partition)$ to denote the set of edges crossing
$\partition$.  
\end{definition}

Finally, we define the notion of \emph{amenability}, which captures
the property of a candidate cycle $C$ having ``low complexity'' with
respect to a partition node. \emph{Friendliness} is the same notion,
but for a near-optimal cycle and an entire \NDHC.

\begin{definition}[Amenability]
  \label{def:amenable}
  For a nonnegative integer $Z$, a cycle $C$ of $G^*$ is
  \emph{$Z$-\amenable} for a partition node $p$ if at most $Z$ of its edges
  cross $\partition(p)$.
  The cycle $C$ is \emph{$Z$-\amenable} for an (ordinary) hierarchical
  clustering $\mathcal T$ if it is $Z$-\amenable for every
  partition node $p$ in $\mathcal T$.
\end{definition}

\begin{definition}[Friendly $\bT$]
  \label{def:friendly}
  For a nonnegative integer $Z$, and $\eps > 0$,
  a $t$-nondeterministic hierarchical clustering $\mathcal T$ of
  $G^*$ is
  \emph{$(Z,\eps)$-\friendly} if there exists a cycle $\widehat C$ in $G^*$ such that:
  \begin{OneLiners}
  \item[(a)] $\widehat C$ has sparsity at most $1+\eps$ times the optimal sparsity for the entire input graph;
  \item[(b)] There exists a forcing $\varphi$ of $\calT$ such that $\widehat C$
    is $Z$-\amenable with respect to ${\mathcal T\vert}_\varphi$.
  \end{OneLiners}
\end{definition}

\subsection{Constructing a Hierarchical Clustering}
\label{sec:procedure}

We now give a procedure $\calA$ that takes integers $t$ and $\Y$,
and constructs a $t$-nondeterministic hierarchical
clustering $\calT$ of part arity~$2\Y$ and depth $O(\log
n)$ and such that each leaf is
  a cluster node $c$ such that $|K(c)|=1$.
In the next section, we show that there is a value of $\Y$ that is
$O(\log n)$ and
a value of $t$ that is $n^{O(\text{polylog } n)}$ for which $\calT$ is
$(Z, \eps)$-friendly with probability at least
$1-\nf1n$.

The \emph{level} $\ell$ of a cluster node $c$ (or partition node $p$)
in $\calT$ is the number of
cluster nodes (or partition nodes) on the path from the root 
to $c$ in $\bT$, not including $c$ (or $p$).  Thus the root and its
child have level~0, its
grandchildren and great-grandchildren have level~1, and so on.
Define
\begin{gather}
  \Delta_{\ell} :=\diameter(G^*)/2^{\ell}.
\end{gather}
The procedure $\calA$ is specified in \Cref{algo:algoA}.
(We later choose the constant $a$ used
in it.) At a high level, it builds $\calT$ top-down.  The root is a
partition node with a single part containing all the vertices and with
a single child, a cluster node whose cluster consists of all the vertices.
The procedure iteratively adds levels to $\calT$.  For each
cluster node $c$ at the current level $\ell$, it randomly selects $O(\log n)$
independent $\Delta_{\ell+1}$-bounded partitions. For each of these
partitions, it considers all possible choices of $2\Y$ parts, and merges the
remaining parts with adjacent chosen parts. (The merging
procedure is given in~\Cref{algo:merge}; note that the merges may cause the
diameter of the parts to increase.) The children of cluster node $c$
are the partition nodes corresponding to these partially-merged
partitions; the children of a partition node $p$ correspond to the
parts of the partition $p$. This process stops when
$\ell = \lceil \log \diameter(G^*)\rceil$, whereupon each non-singleton
cluster is shattered into singleton nodes. Note the two unusual parts
to this construction: the use of multiple partitions for each cluster
(which makes it a \emph{nondeterministic} partition), and the merging of
parts (which bounds the part arity by~$2\Y$, at the cost of increasing the
cluster diameter).
  \begin{algorithm}[h]
    \caption{Procedure $\calA$ to construct the $t$-\NDHC with part
      arity $\Y$} \label{algo:algoA}
    \textbf{init:} $\bT \gets$ root partition node $p$ with $\pi(p) =
    \{V(G^*)\}$, and child $c$ with $K(c) = V(G^*)$\;
    \For{$\ell=0, 1, \ldots, \lceil\log \diameter(G^*)\rceil$}{\label{algo:forell-real}
      \ForEach{cluster node $c$ at level $\ell$ with $|K(c)|>1$}{
        \For{$i=1, \ldots, a\log n$}{\label{algo:foreye-real}
          $\partition_i \gets$ independent random $\Delta_{\ell+1}$-bounded partition of $K(c)$ \;
          \ForEach{subset $\barp\neq \emptyset$ of at most $2\Y$ parts of $\partition_i$}{\label{algo:forkappa-real}
            $\partition_{i,\barp} \gets$
            \textsc{MergeParts}$(\partition_i, \kappa)$ \;
            create a child partition node $p_{i,\barp}$ of $c$ with $\partition(p_{i,\barp}) =
            \partition_{i,\barp}$\;
            \ForEach{part $K$ of $\partition(p_{i,\barp})$}{\label{algo:forK-real}
              create a child cluster node $c'$ of $p_{i,\barp}$ with $K(c')=K$\;
            }
          }
        }
      }
    }
  \ForEach{leaf cluster node $c$ with $|K(c)|> 1$}{
    let $c$ have a single child partition node $p$ which shatters
    $K(c)$: $\partition(p) := \{\{v\} \mid v \in K(c)\}$\;
    let $p$ have a child cluster node $c_v$ for each node $v \in
    K(c)$, having $K(c_v):=\set{v}$ \;
  }
  \end{algorithm}

  \begin{algorithm}[h]
    \caption{\textsc{MergeParts}$(\pi,\barp)$}\label{algo:merge}
    $\pi' \gets \pi$ \;
    \While{$\exists$ edge $uv$ such that the part of $\pi'$ containing
      $u$ intersects $\kappa$ and the part containing $v$ does not
      intersect $\kappa$}{
      merge the two parts}
      \Return $\pi'$
  \end{algorithm}

\begin{lemma}[Properties I]
  The $t$-nondeterministic hierarchical clustering $\calT$ of the dual
  graph $G^*$ produced by \Cref{algo:algoA} has $t =
  n^{O(\Y)}$, depth $O(\log n)$, part arity at
  most $2\Y$, and  total size $n^{O(\Y \log n)}$.
\end{lemma}

\begin{proof}
  The loops in \Cref{algo:foreye-real,algo:forkappa-real} iterate over
  $a \log n$ and $n^{2\Y}$ values respectively; the value of $t$ is at
  most their product. The part arity follows from the fact that each
  partition $\pi_{i,\barp}$ has at most $2\Y$ parts. The depth of
  $\calT$ is at most $O(\log \diameter(G^*)) = O(\log n)$, thanks to
  \Cref{lem:aspectratio}.  Hence the total size of $\calT$ is at most
  $(a \log n \cdot n^{2\Y} \cdot \Y)^{\text{depth}} = n^{O(\Y \log
    n)}$.
\end{proof}

In \S\ref{sec:structure-theorem}, we prove that there is a choice of
$\Y$ with $\Y = O(\log n/\eps)$ for which $\calT$ is
$(Z,\eps)$-friendly with high probability.  In \S\ref{sec:sa} we show how to write and round a linear
program to find an approximate sparsest cut, given a
$(Z,\eps)$-friendly $t$-\NDHC $\calT$.

\def\HiLi{\leavevmode\rlap{\hbox to \hsize{\color{yellow!50}\leaders\hrule height .8\baselineskip depth .5ex\hfill}}}

\section{The Structure Theorem}
\label{sec:structure-theorem}

Before proceeding with the description of the algorithm, we state and prove the structure theorem.

\begin{theorem}
  \label{thm:structure}
  There is a choice of $\Y$ that is $O(\log n/\eps)$
  for which the NDHC produced by procedure $\calA$
  is a  $(\Y,\eps)$-friendly
  with probability at least $1-\nf1n$.
\end{theorem}

To prove \Cref{thm:structure}, we describe and analyze a \emph{virtual
  procedure} $\calV$ given in \Cref{algo:algoV} which performs the
steps from the \emph{actual} procedure $\calA$, plus some extra
\emph{virtual} steps in the background for the purpose of the
analysis. It takes as input not only the graph $G^*$ but also a cycle $C_0$ with
optimal sparsity (which the algorithm $\calA$ does not have).  The procedure
might fail but we show that failure occurs with probability at most $1/n$.
Assuming no failures, it produces not only $\calT$ but also a polynomial-size set
$\calC$ of cycles such that
\begin{OneLiners}
\item[(a)] for each cycle $C\in \calC$, there is a forcing $\varphi$ of $\bT$ such that $C$ is
  $\Y$-amenable with respect to ${\mathcal T\vert}_\varphi$.
\item[(b)] The total cost of cycles in $\calC$ is at most $1+\eps$
  times the cost of $C_0$.
\item[(c)] Every two faces of $G^*$ separated by $C_0$ are separated
  by some cycle in $\calC$.
\end{OneLiners}
Lemma~\ref{lem:good-cycle} then implies that there is a cycle
$\widehat C$ whose sparsity is at most $1+\eps$ times optimal and
that is $\Y$-amenable with respect to ${\mathcal
  T\vert}_\varphi$ for some forcing $\varphi$ of $\bT$.
\begin{explanation}
Again, we emphasize that the virtual procedure $\calV$ is merely a
thought experiment for the analysis.
The lines added to $\calA$ are
highlighted for convenience.
We now outline $\calV$, and state the lemmas
that prove \Cref{thm:structure}.  We then prove these lemmas in
\S\ref{sec:proof-struct}.
\end{explanation}

\subsection{The Virtual Procedure \texorpdfstring{$\calV$}{calV}}
\label{sec:virtual}

  \begin{algorithm}[h]
  \caption{Virtual Procedure $\calV$} \label{algo:algoV}
  \textbf{init:} $\bT \gets$ root partition node $p$ with $\pi(p) =
  \{V(G^*)\}$, and child $c$ with $K(c) = V(G^*)$\;
  \HiLi \textbf{init:} $\calC \gets \{C_0\}$,
  $\psi[\cdot,\cdot] \gets$ an empty table \;
  \For{$\ell=0, 1, \ldots, \lceil \log \diameter(G^*)\rceil$}{\label{algo:forell-virtual}
    \ForEach{cluster node $c$ at level $\ell$ with $|K(c)|>1$}{
      \HiLi
       \lForEach{cycle $C \in \calC$ valid for cluster
         $c$}{ $\calC \gets (\calC \setminus \{C\})\, \cup $ 
         \textsc{Patch}$(C,c,\Delta_\ell)$}
      \For{$i=1, \ldots, a\log n$}{\label{algo:foreye-virtual}
        $\partition_i \gets$ random $\Delta_{\ell+1}$-bounded partition of $K(c)$ \;
        \ForEach{subset $\barp\neq \emptyset$ of at most $2\Y$ parts of $\partition_i$}{\label{algo:forkappa-virtual}
          $\partition_{i,\barp} \gets$
          \textsc{MergeParts}$(\partition_i, \kappa)$ \;
          create a child partition node $p_{i,\barp}$ of $c$ with $\partition(p_{i,\barp}) =
          \partition_{i,\barp}$ \label{algo:MergeParts-step} \;
          \ForEach{part $K$ of $\partition(p_{i,\barp})$}{\label{algo:forK-virtual}
            create a child cluster node $c'$ of $p_{i,\barp}$ with $K(c')=K$\;
          }
        }
      }
      \HiLi \ForEach{cycle $C\in \calC$ that is valid for cluster
        node $c$}{\label{algo:forcycle-virtual}
        \HiLi \lIf{$|\{e \in C \mid e \text{ crosses } \pi_i \}| >
          \Y$ for all $i \in [a \log n]$}{\textbf{fail}} \label{algo:crossing-check}
        \HiLi $i \gets$ any index in $[a \log n]$ s.t.\  $|\{e \in C
        \mid e \text{ crosses } \pi_i \}| \leq
        \Y$ \label{algo:choice-of-psi}\;
        \HiLi $\barp \gets \{ P \in \pi_i \mid P$  contains an endpoint
        of an edge of $ C$
        internal to  $K(c)\}$ \label{algo:choice-of-kappa}\;
        \HiLi \lIf{$\barp = \emptyset$}{$\barp \gets \set{P}$ where
          $P$ is chosen arbitrarily from $\pi_i$}
        \HiLi $\psi[c,C] \gets p_{i,\barp}$ \label{algo:table-update} \;
      }
    }
  }
  \ForEach{leaf cluster node $c$ with $|K(c)|> 1$}{
    let $c$ have a single child partition node $p$ which shatters
    $K(c)$: $\partition(p) := \{\{v\} \mid v \in K(c)\}$\label{algo:shatter-partition-node} \;
    \HiLi $\psi[c,C] \gets p$ for each cycle $C$ in $\calC$ \label{algo:shatter-table-update} \;
    let $p$ have a child cluster node $c_v$ for each node $v \in
    K(c)$, i.e., having $K(c_v):=\set{v}$ \;
  }
\end{algorithm}

The virtual procedure $\calV$, in addition to building the tree
$\calT$, maintains a set $\calC$ of cycles.  Initially, $\calC$
consists of a single cycle, the optimal solution $C_0$. When
processing a cluster, $\calV$ first calls a procedure \textsc{Patch}
on each cycle ``valid'' for the cluster. (We define the notion of
validity soon.) %
This patching possibly replaces the cycle $C$ with several cycles
which jointly include all the edges of $C$ and possibly some
additional edges internal to the cluster in such a way that
two faces of $G^*$ separated by $C$
are also separated by at least one of the replacement
cycles.  The
cycles $C$ considered by $\calV$ thus form a rooted tree (a
``tree of cycles''), where the children of a cycle $C$ are the cycles
that replace it in $\calC$. The patching process ensures that the
total cost of cycles in $\calC$ remains small.

\begin{restatable}{lemma}{Separation} \label{lem:separation}
Throughout the execution of the
  virtual procedure, every two faces of $G^*$ separated by $C_0$ are
  separated by some cycle in $\calC$.
\end{restatable}

\begin{restatable}{lemma}{SumOfCosts}
  \label{lem:sum-of-lengths}
  When virtual procedure $\calV$ terminates, the sum of costs of cycles in
  $\calC$ is at most $1+\eps$ times the cost of $C_0$.
\end{restatable}

The cost of the starting simple cycle $C_0$ is at most $O(n^3)$, since
all edges have cost at most $n^2$ (by~\Cref{lem:aspectratio}). Hence
the final cost remains $O(n^3)$, for constant $\eps$. However, since
each edge has cost at least $1$, each final cycle has cost at least
$1$,  which gives the following corollary.

\begin{corollary}
  \label{cor:quadratic}
  For any constant $\eps > 0$, the collection $\calC$ at any
  point during the virtual algorithm's run contains at most $O(n^3)$
  cycles.
\end{corollary}

\Cref{lem:separation} allows us to apply \Cref{lem:good-cycle} to the
cycles comprising $\calC$.  By 
\Cref{lem:sum-of-lengths}, at least one of these cycles has
near-optimal sparsity, proving the first condition from
\Cref{def:friendly}. To prove the second condition, we need to show
that such a cycle also has a forcing with good amenability. Indeed,
the reason for patching the cycles was to ensure that the cost of
edges of this cycle internal to cluster $c$ is not much more than the
$\Delta_\ell$ value, so that at least one of the random partitions
crosses it only a few times (with high probability).

For a nonleaf cluster node $c$ of $\calT$ and a cycle $C$, recall from
\Cref{algo:algoV} that $\psi[c,C]$ is assigned one of the children of
$c$.  (The child is necessarily a partition node because the parent is
a cluster node.)  This is the way that the virtual algorithm records
which of the random partitions is crossed only a few times by a given
cycle $C$.  In case the cycle $C$ or a cycle derived from $C$ (i.e. a
descendant in the tree of cycles) turns out to be nearly optimal, we
want there to exist a forcing that induces an ordinary hierarchical
decomposition from $\bT$ such that $C$ or its descendant crosses the
partition of every partition node remaining in the ordinary
hierarchical decomposition.  The following definition captures the
idea that, for a given cycle $\bar C$, a given node $x$ remains in the
ordinary hierarchical decomposition corresponding to $\bar C$: it
states that, for every cluster node $c$ that is a proper ancestor of
$x$, for an appropriate ancestor $C$ of $\bar C$, $\psi[c,C]$ ``points
to'' the child of $c$ that is an ancestor of $x$.

\begin{definition}[Validity]
  \label{def:validity}
  A cycle $\bar C$ is valid for a node $x$ in $\calT$ if for each
  cluster node $c$ that is a proper ancestor of $x$ in
    $\calT$, there is an ancestor cycle $C$ of $\bar C$ in the tree
  of cycles such that $\psi[c,C]$ is an ancestor of $x$.
\end{definition}
Note that the initial cycle $C_0$ is trivially valid for root node
$p_0$ and its child.

\begin{restatable}{lemma}{ForcingDef}
  \label{lem:forcing-def}
  Suppose the algorithm $\calV$ completes without failure. Then for each $C \in
  \calC$, there is a forcing $\varphi_C$ such that $C$ is
  $\Y$-amenable with respect to $\calT|_{\varphi_C}$.
\end{restatable}

\begin{restatable}{lemma}{FailureProb}
\label{lem:failure-probability} The probability of any
  failure occurring during $\calV$'s execution is at most
  $\nf1n$.
\end{restatable}

This proves the second condition of \Cref{def:friendly} for every
cycle in the collection $\calC$, and hence for the near-optimal cycle
inferred in~\Cref{lem:sum-of-lengths}. Hence we have proved
\Cref{thm:structure}, modulo the proofs of the lemmas above. We now
describe the subprocedure {\sc Patch} and then give the proofs.

\subsection{The Subprocedure {\sc Patch}}
\label{sec:patching}

A call to subprocedure {\sc Patch} takes three arguments: 
(i)~a cycle $C$,
(ii)~a
cluster node
$c$,
 and (iii)~a cost
 $\Delta_\ell$. It outputs a collection of cycles.

Note that $c$ corresponds to a subset $K(c)$ of vertices of $G^*$, and
that some edges of the cycle $C$ might not be internal to $K(c)$.

\subsubsection{Steps of the Subprocedure}

The subprocedure is as follows. If
$\cost(C\cap G^*[K(c)])$ is at most
$(\Y/3) \cdot \Delta_\ell$, then it returns the set consisting solely
of $C$.  

Otherwise, as illustrated in
\Cref{fig:cluster-and-cycle,fig:cluster-and-cycle-two},
it initializes a counter to 0, selects an arbitrary starting vertex of
$C$ that is in $K(c)$, %
which it designates a \emph{special} vertex, then traverses $C$ in an
arbitrary direction from the starting vertex. Each time the procedure
traverses an  edge from $u$ to
$v$ that is internal to $K(c)$
it increments the counter by $\cost(uv)$.  If the resulting
value of the counter exceeds $(\Y/3) \cdot \Delta_\ell$, the procedure
designates the edge $uv$ as a special edge, designates the
vertex $v$ as a special vertex, and resets the counter to 0. It then
continues this traversal of $C$ from $v$ onwards, creating special edges
and nodes, etc., until it returns to the starting vertex.

Let $r$ be the starting vertex.  The procedure now selects shortest
paths in the graph induced by $K(c)$ from $r$ to each of the other
special vertices. Let $\calP$ be the multiset of such shortest paths,
where each selected path is included with multiplicity two.

Now the union of $C$ and the paths of $\calP$ is decomposed into
cycles.  For $t=1,2, \ldots $, a cycle is formed consisting of the
path from the center to the $t^{th}$ special vertex, along $C$ to the
$t+1^{st}$ special vertex, and then back along the path to the center.
A final cycle is formed consisting of the path from $r$ to
the last special vertex and along $C$ to $r$. The set consisting of
these cycles is returned by the procedure.

 We refer to the above steps that replace $C$ by
this collection as \emph{patching} the cycle $C$.

\begin{figure} \centering
\begin{subfigure}[t]{.3\textwidth} \centering
  \includegraphics[width=1.5in,page=1]{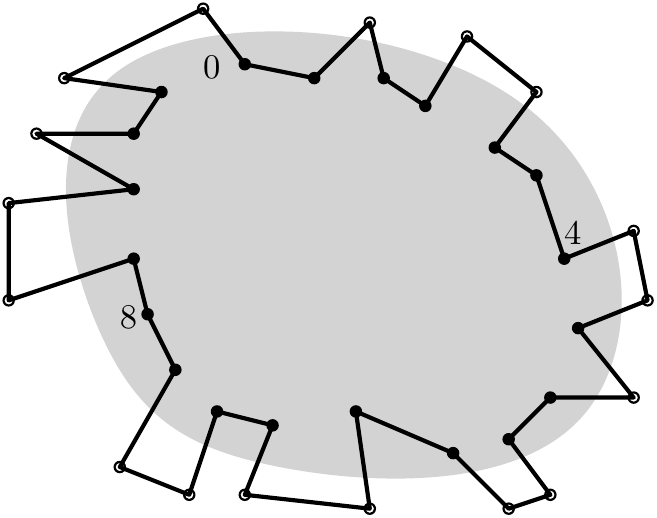}
\end{subfigure}
\quad
\begin{subfigure}[t]{.3\textwidth} \centering
  ~~~\includegraphics[width=1.5in,page=2]{figures/newpatch.pdf}
\end{subfigure}
\quad
\begin{subfigure}[t]{.3\textwidth} \centering
  ~~~\includegraphics[width=1.5in,page=3]{figures/newpatch.pdf}
\end{subfigure}
\caption{(a) A cluster and a cycle whose cost inside the cluster
  is much larger than the scale.
  The numbers indicate the counter values on the cost of the cycle at that vertex. (b)
  Two paths added to each special vertex, the directions are only for
  emphasis. (c) A decomposition into three cycles.}
\label{fig:cluster-and-cycle}
\end{figure}

\subsubsection{Properties}

\begin{lemma}
  \label{lem:Patch-property} In a call to {\sc
    Patch}$(C, c, \Delta_{\ell})$, each cycle formed by patching
  $C$ consists of edges in $C$ and edges internal to
  $K(c)$.
\end{lemma}

\begin{proof} Each edge in such a cycle
  that is not in the original cycle $C$ is in one of the
  shortest paths in the graph induced by $K(y)$.
\end{proof}

\begin{corollary} \label{cor:Patch-property}
For any partition node $p$ that is not a descendant
  of $c$, if $C$ is $\Y$-amenable for $p$ then
  so is every cycle obtained by patching it.
\end{corollary}
 
\begin{proof}
  By $\Y$-amenability, at most $\Y$ edges of $C$ cross
  $\partition(p)$.  Let $C'$ be a cycle formed by patching $C$.  By
  Lemma~\ref{lem:Patch-property}, an edge in $C' - C$ is internal to
  $K(c)$, so it cannot cross $\partition(p)$.  Thus at most $\Y$
  edges of $C'$ cross $\partition(p)$.
\end{proof}

\begin{claim}
  Consider the call {\sc  Patch}$(C, c, \Delta_{\ell})$.
  The cost of the shortest path from $r$ to each special vertex is
  at most $\Delta_\ell$.
\end{claim}

\begin{proof}
  Let $p$ be the parent of $c$.  Assume $p$ is not the root, and let
  $c'$ be the parent of $p$.  Because $C$ is valid for $c$ (a
  precondition for the call), for every proper ancestor of $c$ that is
  a cluster node, and in particular for the grandparent $c'$ of $c$,
  and for some ancestor $C'$ of $C$.
  $\psi[c',C']$ is an ancestor of $c$.  Because $\psi[c',C']$ is a
  child of $c'$ and an ancestor of $c$, it must be the parent $p$ of $c$.
  Therefore, in some execution of
  line~\eqref{algo:table-update}, $p=p_{i, \kappa}$ where
\begin{OneLiners}
    \item $\pi_i$ is a
      $\Delta_{\ell}$-bounded partition of $K(c')$,
      \item $\kappa$ is the
        set of parts of $\pi$ that intersect $C$,
      \item $\pi_{i,\kappa}=
        \textsc{MergeParts}(\pi_i, \kappa)$, and
      \item $\pi(p_{i,\kappa})=\pi_{i.\kappa}$.
\end{OneLiners}        
If $C$ has no vertices in $K(c)$ then $\cost(C\cap G^*[K(c)])$ is zero so
\textsc{Patch}$(C,c,\delta_\ell)$ does not change $C$.  Suppose $C$
has some vertex $v$ in $K(c)$. Because $c$ is a child of $p$, $K(c)$
is a part of $\pi(p_{i,\kappa}) = \pi_{i,\kappa}$.  
Each part of $\pi_{i,\kappa}$ is obtained by merging a part of $\pi_i$
belonging
to $\kappa$ with some parts of $\pi_i$ not in $\kappa$.  Let $S$ be the
part of $\pi_i$ that contains $v$.  Because $S$ is the only part
belonging to $\kappa$ that is a subset of $K(c)$, every vertex in
$C\cap K(c)$ belongs to $S$.  In particular, all the vertices
designated as special in \textsc{Patch} belong to $S$.  Because $S$ is
a part of $\pi_i$, it has diameter at most $\Delta_\ell$.  This proves
the claim.
\end{proof}

\begin{corollary} \label{fact:cycle-cost} If cycle $C'$ is one of the
  cycles produced by patching cycle $C$
  then the sum of costs of non-special edges in $C'\cap G^*[K(c)]$ is
  at most $(\frac{\Y}{3} + 2)\cdot \Delta_\ell$.
\end{corollary}

\begin{corollary} \label{cor:total-cost-patching-one-cycle}
  The total cost of edges that are in cycles formed by patching $C$
  and are not in $C$, taking into account multiplicity, is at most
  $$2\,\left(\frac{\cost(C\cap G^*[K(c)])}{(\Y/3)\cdot \Delta_\ell}\right) 2\Delta_\ell.$$
\end{corollary}

\begin{proof}  The number of shortest paths added
  is twice the number of special vertices.  For each special
  vertex other than the starting vertex, the algorithm scans a portion
  of $C\cap G^*[K(c)]$ of cost greater than $(\Y/3)\cdot \Delta_\ell$.
  Let $\eta$ be the smallest cost scanned.  The number of
  special vertices is at most
  $$1 + \left\lfloor\frac{\cost(C\cap
      G^*[K(c)])}{\eta}\right\rfloor$$
  which is at most
  $$\left\lceil\frac{\cost(C\cap G^*[K(c)])}{(\Y/3)\cdot \Delta_\ell}\right\rceil.$$
which is at most $2\frac{\cost(C\cap G^*[K(c)])}{(\Y/3)\cdot \Delta_\ell}$
because $\cost(C\cap G^*[K(c)])> (\Y/3) \cdot \Delta_\ell$.
\end{proof}

\begin{figure} \centering
\begin{subfigure}[t]{.3\textwidth} \centering
  \includegraphics[width=2in,page=1]{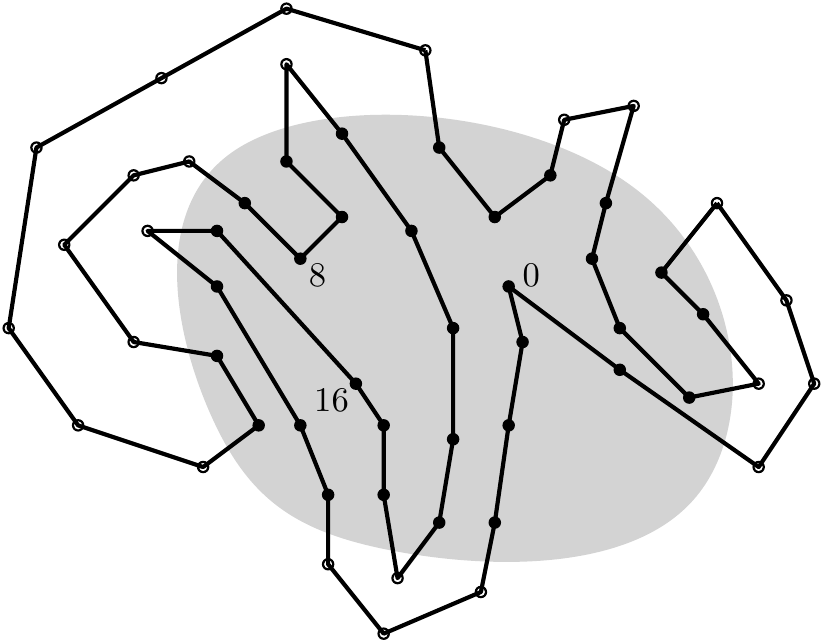}
\end{subfigure}
\quad
\begin{subfigure}[t]{.3\textwidth} \centering
  ~~~\includegraphics[width=2in,page=2]{figures/newpatch2.pdf}
\end{subfigure}
\quad
\begin{subfigure}[t]{.3\textwidth} \centering
  ~~~\includegraphics[width=2in,page=3]{figures/newpatch2.pdf}
\end{subfigure}
\caption{Another example of patching, with a more complicated cycle;
  this time the special vertices correspond to vertices where the
  counter is a multiple of $8$.}
\label{fig:cluster-and-cycle-two}
\end{figure}

\subsection{The remaining proofs}
\label{sec:proof-struct}

Now we restate and prove the remaining lemmas from \S\ref{sec:virtual}.

\Separation*

\begin{proof} The patching procedure applied to a cycle $C$ ensures
  that each edge appears an odd number of 
  times in the resulting cycles (counting multiplicity) iff the edge
  belongs to $C$.  It follows from
  \Cref{lem:preserve-demand-separation} that every two faces separated by
  $C$ are separated by at least one of the cycles resulting from
  patching.  The lemma then follows by induction.
\end{proof}

\SumOfCosts*

\begin{proof}[Proof of \Cref{lem:sum-of-lengths}]
  Consider iteration $\ell$ of the virtual procedure, which operates
  on cluster nodes $c$ having level $\ell$. Let $\calC_\ell$ denote
  the set $\calC$ at the end of this iteration. For a cycle
  $C \in \calC_\ell$, consider the cluster nodes at this level.  The  clusters $K(c)$
  corresponding to these nodes are disjoint, and hence edges internal
  to these clusters are disjoint.

  Now let $c$ be one such cluster node. For cycle $C \in \calC_\ell$,
  define $\calC(c,C)$ as the set of cycles $C'$ that are ancestors of
  $C$ that were valid and patched \emph{at the time $c$ was
    processed} by the virtual procedure $\calV$. For each cycle
  $C'\in \calC(c,C)$, \Cref{cor:total-cost-patching-one-cycle} bounds the increase in
  total cost due to patching $C'$. Hence, the total increase during
  iteration $\ell$ is at most
  \begin{gather*}
    \sum_{C \in
      \calC_\ell} \quad   \sum_{c \text{ at level } \ell} \quad \sum_{C' \in \calC(c,C)}
    2\,\left(\frac{\cost(C'\cap G^*[K(c)])}{(\Y/3)\cdot \Delta_{\ell}} \right)
    2\Delta_{\ell}\\
    \leq \frac{12}{\Y}\sum_{C \in \calC_\ell} \cost(C)
  \end{gather*}
  where we use the fact that for each cycle $C \in \calC_\ell$, the relevant
  valid clusters are disjoint.
  
We have shown that iteration $\ell$ increases the cost by at most a
factor of $1+\nicefrac{12}{\Y}$.  Because the number of iterations is
$O(\log n)$, we can choose $\Y = O(\eps^{-1} \log n)$ so that the
total increase over all iterations is at most $1+\eps$.  %
\end{proof}

\ForcingDef*

\begin{proof}[Proof of~\Cref{lem:forcing-def}]
  Let $\bar C$ be a cycle in the final collection $\calC$.  The nodes of
  $\calT$ for which $\bar C$ is valid form an ordinary hierarchical partition of
  $G^*$.  Define $\varphi_{\bar C}$ as the function that maps each such cluster node $c$ to
  its unique partition node child $p$ in this ordinary hierarchical
  partition; validity implies that $\psi[c,C] = p$ for an ancestor
  $C$ of $\bar C$. 

  Let $p$ be a partition node of $\calT|_{\varphi_{\bar C}}$.
  Let $c$ be the parent of $p$.  By definition of validity, 
  $\psi[c, C]$ was assigned $p$ for some ancestor $C$ of $\bar
  C$. 

  Case 1: $\psi[c,C]$ was assigned $p$ in some execution
  of~\Cref{algo:table-update}. In this case, $\pi(p) = \pi_{i,\kappa}$
  where $i$ is selected in~\Cref{algo:choice-of-psi} and $\kappa$ is
  selected in~\Cref{algo:choice-of-kappa} and $\pi_{i,\kappa}$ is
  derived in~\Cref{algo:MergeParts-step} by application of
  \textsc{MergeParts} to $\pi_i$ and $\kappa$.  By choice of $i$
  in~\Cref{algo:choice-of-psi},  at most $\Y$ edges of $C$ cross
  $\pi_i$, so the same holds for $\pi_{i,\kappa}$.

  Case 2: $\psi[c,C]$ was assigned $p$ in some execution
  of~\Cref{algo:shatter-table-update}.  In this case, $c$ is a part
  of a partition $\pi_{i,\kappa}$ selected in iteration $\ell=\log
  {\diameter}(G^*)$, where $i$ is selected in~\Cref{algo:choice-of-psi} and $\kappa$ is
  selected in~\Cref{algo:choice-of-kappa} and $\pi_{i,\kappa}$ is
  derived in~\Cref{algo:MergeParts-step} by application of
  \textsc{MergeParts} to $\pi_i$ and $\kappa$. Because $\pi_i$ is a
  $\Delta_{\log {\diameter}(G^*)+1}$-bounded partition,
  each part of $\pi_i$ has diameter less than one, but edge-lengths are integral
  and positive, so each part consists of a single vertex.  By
  choice of $\kappa$ and definition of \textsc{MergeParts}, each part
  of $\pi_{i,\kappa}$ contains at most one part $P$ of $\pi_i$ that
  contains an endpoint of an edge of $C$ internal to $K(c)$, which
  implies that no edge of $C$ is internal to a part of
  $\pi_{i,\kappa}$.  Therefore $C$ is 0-amenable for the shattering
  partition node of~\Cref{algo:shatter-partition-node}.
\end{proof}

\FailureProb*

\begin{proof}[Proof of Lemma~\ref{lem:failure-probability}]
  Consider iteration $\ell$ of the virtual
  procedure. Let $\calC_\ell$ denote the set $\calC$ at the end of
  this iteration. For a cycle $\bar C \in \calC_\ell$, consider the
  cluster nodes at this level for which $C$ is valid.
  The clusters $K(c)$ corresponding to these nodes $c$ are disjoint,
  and hence there are at most $n$ such clusters.

  For such a cluster node $c$, consider the ancestor $C$ of $\bar C$
  such that $\psi[c,C]$ is assigned a partition node in \Cref{algo:table-update}. Since $C$ was
  the result of patching a cycle with respect to $c$,
 \Cref{fact:cycle-cost} implies that
  the cost of non-special edges in $C$ is at most
  $(\nf{\Y}{3} +2)\Delta_\ell$. By the properties of random
  $\Delta$-bounded partitions, the expected number of nonspecial edges
  of $C$ crossing the random $\Delta_{\ell+1}$-bounded partition is
  bounded by $(\Y/3+2)\cdot \frac{\Delta_{\ell}}{\Delta_{\ell+1}}$,
  which is $2(\Y/3 + 2)$.  Counting the single special edge in the
  cycle, the expected number of edges internally crossing the
  partition is at most $2\Y/3 + 5$.  Therefore, by Markov's
  inequality, the number is at most $\Y$ with some constant
  probability.

  Thus over all $O(\log n)$ iterations of the for-loop
  in~\Cref{algo:foreye-virtual}, the probability that none of them
  selects a $\Delta_{\ell+1}$-bounded partition with at most $\Y$
  edges of $C$ internally crossing the partition is at most
  $\exp\{O(a \log n)\} = 1/n^{O(a)}$. Now we can take a naive union
  bound over all the cluster nodes $c$ at this level for which $\bar
  C$ is valid (of which there are at most $n$, since
  they are disjoint), over all cycles $\bar C \in \calC_\ell$, and over all
  levels $\ell$. These are only $O(n^3 \log n)$ events (see~\Cref{cor:quadratic}), so choosing
  $a$ to be a sufficiently large constant ensures that the success
  probability of the virtual procedure is at least $1-1/n$, as
  claimed.
\end{proof}

{%

\newcommand{\f}{\frac}
\newcommand{\cd}{\cdot}
\newcommand{\bn}{\binom}
\newcommand{\sr}{\sqrt}
\newcommand{\cds}{\cdots}
\newcommand{\lds}{\ldots}
\newcommand{\vds}{\vdots}
\newcommand{\dds}{\ddots}
\newcommand{\pge}{\succeq}
\newcommand{\ple}{\preceq}
\newcommand{\sm}{\setminus}
\newcommand{\s}{\subseteq}
\newcommand{\su}{\supseteq}
\renewcommand{\emptyset}{\varnothing}

\newcommand{\sumni}{\sum_{n=1}^\infty}
\newcommand{\sumin}{\sum_{i=1}^n}
\newcommand{\bigcupni}{\bigcup_{n=1}^\infty}
\newcommand{\bigcupin}{\bigcup_{i=1}^n}
\newcommand{\bigcapni}{\bigcap_{n=1}^\infty}
\newcommand{\bigcapin}{\bigcap_{i=1}^n}

\newcommand{\BE}{\begin{enumerate}}
\newcommand{\EE}{\end{enumerate}}
\newcommand{\im}{\item}
\newcommand{\BI}{\begin{itemize}}
\newcommand{\EI}{\end{itemize}}
\def\BAL#1\EAL{\begin{align*}#1\end{align*}}
\def\BALN#1\EALN{\begin{align}#1\end{align}}
\def\BG#1\EG{\begin{gather}#1\end{gather}}

\newcommand{\Sum}{\displaystyle\sum\limits}
\newcommand{\Prod}{\displaystyle\prod\limits}
\newcommand{\Int}{\displaystyle\int\limits}
\newcommand{\Lim}{\displaystyle\lim\limits}
\newcommand{\Max}{\displaystyle\max\limits}
\newcommand{\Min}{\displaystyle\min\limits}

\newcommand{\logn}{\log n}

\newcommand{\dx}{\frac d{dx}}
\newcommand{\dy}{\frac d{dy}}
\newcommand{\dz}{\frac d{dz}}
\newcommand{\dt}{\frac d{dt}}

\newcommand{\inv}{^{-1}}

\newcommand{\R}{\mathbb R}
\newcommand{\Z}{\mathbb Z}
\newcommand{\N}{\mathbb N}
\newcommand{\Q}{\mathbb Q}

\newcommand{\de}{\delta}
\newcommand{\De}{\Delta}
\newcommand{\la}{\lambda}
\newcommand{\g}{\gamma}
\newcommand{\pt}{\partial}
\newcommand{\al}{\alpha}
\newcommand{\be}{\beta}
\newcommand{\om}{\omega}
\newcommand{\Om}{\Omega}
\newcommand{\el}{\ell}
\renewcommand{\th}{\theta}
\newcommand{\Th}{\Theta}
\newcommand{\m}{\mathcal}

\newcommand{\Ra}{\Rightarrow}

\newcommand{\lf}{\lfloor}
\newcommand{\rf}{\rfloor}
\newcommand{\lc}{\lceil}
\newcommand{\rc}{\rceil}

\newcommand{\Var}{\textup{Var}}
\newcommand{\Cov}{\textup{Cov}}
\newcommand{\1}{\mathbbm 1}
\newcommand{\polylog}{\textup{polylog}}
\newcommand{\pl}{\textup{polylog}}
\newcommand{\norm}[1]{\left\lVert#1\right\rVert}
\newcommand{\vol}{\textbf{\textup{vol}}}

\newcommand{\rank}{\textup{rank}}
\newcommand{\spn}{\textup{span}}
\newcommand{\Tr}{\textup{Tr}}

\newcommand{\lp}{\left(}
\newcommand{\rp}{\right)}
\newcommand{\lb}{\left[}
\newcommand{\rb}{\right]}
\newcommand{\lmt}{\left[\begin{matrix}}
\newcommand{\rmt}{\end{matrix}\right]}

\newcommand{\BT}{\begin{theorem}}
\newcommand{\ET}{\end{theorem}}
\newcommand{\BL}{\begin{lemma}}
\newcommand{\EL}{\end{lemma}}
\newcommand{\BD}{\begin{definition}}
\newcommand{\ED}{\end{definition}}
\newcommand{\BC}{\begin{corollary}}
\newcommand{\EC}{\end{corollary}}
\newcommand{\BO}{\begin{observation}}
\newcommand{\EO}{\end{observation}}
\newcommand{\BCL}{\begin{claim}}
\newcommand{\ECL}{\end{claim}}
\newcommand{\BSCL}{\begin{subclaim}}
\newcommand{\ESCL}{\end{subclaim}}
\newcommand{\BF}{\begin{fact}}
\newcommand{\EF}{\end{fact}}
\newcommand{\BA}{\begin{assumption}}
\newcommand{\EA}{\end{assumption}}
\newcommand{\BP}{\begin{proof}}
\newcommand{\EP}{\end{proof}}
\newcommand{\BSP}{\begin{subproof}}
\newcommand{\ESP}{\end{subproof}}
\newcommand{\BPS}{\begin{proof}[Proof (Sketch)]}
\newcommand{\EPS}{\end{proof}}
\Crefname{observation}{Observation}{Observations}
\Crefname{claim}{Claim}{Claims}
\Crefname{subclaim}{Subclaim}{Subclaims}
\Crefname{fact}{Fact}{Facts}
\Crefname{assumption}{Assumption}{Assumptions}

\newenvironment{subproof}[1][\proofname]{%
  \renewcommand{\qedsymbol}{$\diamond$}%
  \begin{proof}[#1]%
}{%
  \end{proof}%
}

\newcommand{\para}{\paragraph}

\newcommand{\tO}{\tilde{O}}

\newcommand{\thml}[1]{\label{thm:#1}}
\newcommand{\thm}[1]{\Cref{thm:#1}}
\newcommand{\leml}[1]{\label{lem:#1}}
\newcommand{\lem}[1]{\Cref{lem:#1}}
\newcommand{\defnl}[1]{\label{def:#1}}
\newcommand{\defn}[1]{\Cref{def:#1}}
\newcommand{\clml}[1]{\label{clm:#1}}
\newcommand{\clm}[1]{\Cref{clm:#1}}
\newcommand{\corl}[1]{\label{cor:#1}}
\newcommand{\cor}[1]{\Cref{cor:#1}}
\newcommand{\obsl}[1]{\label{obs:#1}}
\newcommand{\obs}[1]{\Cref{obs:#1}}
\newcommand{\eqnl}[1]{\label{eq:#1}}
\newcommand{\eqn}[1]{(\ref{eq:#1})}
\newcommand{\linel}[1]{\label{line:#1}}
\renewcommand{\line}[1]{line~\ref{line:#1}}
\newcommand{\secl}[1]{\label{sec:#1}}
\renewcommand{\sec}[1]{\Cref{sec:#1}}

\newcommand{\inside}{\textsf{inside}}
\newcommand{\fpath}{\textsf{partn\_path}}
\newcommand{\parent}{\textsf{par}}
\newcommand{\GG}{\overline{\Gamma}}

\section{Finding a Sparse Cut via LPs}
\label{sec:sa}

The structure theorem (\Cref{thm:structure}) from
\S\ref{sec:structure-theorem} gives us a $t$-non\-deterministic
hierarchical clustering $\calT$ of the dual graph $G^*$ that is
$(Z,\e)$-friendly with high probability.
Given such a clustering $\calT$, we need to find a good forcing and a good ``low-complexity''
solution with respect to it. The natural approach to try is 
dynamic programming, but no such approach is currently known for
non-uniform sparsest cut; for example, the problem is NP-hard even for
treewidth-$2$ graphs. Hence we solve a linear program and
round it.
Our linear program is directly inspired by the Sherali-Adams approach
for bounded-treewidth graphs, augmented with ideas specific to  the planar
case. It encodes a series of choices giving us a forcing, and also
choices about how a $(1+\eps)$-approximate solution crosses the
resulting clustering. Of course, these choices are fractional, so we
need to round them, which is where we lose the factor of $2$.

\subsection{Notation}

As in
the previous sections, we work  on the planar dual
$G^*=(V^*,E^*)$. Given a partition $\pi$ of a vertex set $U\s V^*$, let
$\partial(\pi)$ denote the \emph{boundary}, i.e., the faces in $G^*$
whose vertices are all in $U$ but belong to more than one part of this partition; see
Figure~\ref{fig:Ap}.\footnote{This boundary $\partial(p)$ is a
  collection of faces of the dual, and hence differs from the typical
  edge-boundary.} Recall that for a partition node $p$ with parent cluster node $c$, we defined $\pi(p)$ as a partition of $K(c)$. In this section, we would like to extend the partition to all of $V^*$ by looking at the partition nodes that are ancestors of $p$. For a (cluster or partition) node $a$ in the tree,
define $\fpath(a)$ as the partition nodes on the path from the root of $\m T$ to $a$,
inclusive.

For a partition node $p$ with parent cluster node $c$, define $\pi^+(p)$ as the following partition of $V^*$: its parts are the parts of $\pi(p)$, together with all parts $P$ in $\pi(p')$ over all partition nodes $p'\in \fpath(p)$, \emph{except} the parts $P$ that contain $K(c)$. By the hierarchical nature of $\m T$, these parts form a valid partition of $V^*$, which we call $\pi^+(p)$. For ease of notation, we abbreviate $\pt(p):=\pt(\pi(p))$ and $\pt^+(p):=\pt(\pi^+(p))$. While we will not need it, the reader can verify that $\pt^+(p)=\biguplus_{p'\in\fpath(p)}\pt(p')$, where the $\biguplus$ indicates disjoint union.%

For any simple 
cycle $C$ in $G^*$, let $\inside(C)$ be the
set of faces of $G^*$ (corresponding to vertices of $G$) inside $C$. Informally, we now define $\m A^+(p)$ as the collection of all subsets of $\pt^+(p)$ that can
comprise the faces of $\pt^+(p)$ inside any cycle $C$ in $G^*$ that is $Z$-amenable for each partition node $p'\in\fpath(p)$. %
Formally, $A^+(p)$ is the set
\BAL
\m \{ \inside(C) \cap \pt^+(p) \mid C \text{ is } & \text{$Z$-\amenable
      for all } p'\in\fpath(p) \}.
\EAL

The following lemma lies at the heart of our ``low-complexity'' argument: the size of $\m A^+(p)$ is quasi-polynomially bounded.

\begin{lemma}\leml{size-A}
Let $\m T$ be $(Z,\e)$-friendly with depth $H$. Then,
  $|\m A^+(p)| \le O(n)^{(H+1)Z}$, and we
  can compute it in time
  $O(n)^{O(H+1)Z}$.
\end{lemma}

\begin{proof}%
Consider a cycle $C$ that is $Z$-amenable for each normal partition
node $p'\in\fpath(p)$ and $0$-amenable for each shattering partition
node $p'\in\fpath(p)$. We first claim that $C$ must cross $\pi^+(p)$
at most $(H+1)Z$ times. Consider an edge $e$ in $C$ whose endpoints
belong to different parts of $\pi^+(p)$. By construction of $\pi^+(p)$
and the $0$-amenability of $Z$ for shattering partition nodes, there
must exist a normal partition node $p'\in\fpath(p)$ such that the endpoints of $e$ belong to different parts of $\pi(p')$. In other words, we can charge each crossing of $\pi^+(p)$ to a crossing of some $\pi(p')$ for some normal partition node $p'\in\fpath(p)$. Since $C$ is $Z$-amenable for normal partition node $p'\in\fpath(p)$, each one can be charged at most $Z$ times, and $|\fpath(p)|\le H+1$, so there are $(H+1)Z$ total crossings of $\pi^+(p)$.

Next, observe that if two cycles $C,C'$ cross the same edges in $\pi^+(p)$, then $\inside(C) \cap \pt^+(p)= \inside(C') \cap \pt^+(p)$, since they can only differ ``within'' parts of $\pi^+(p)$. It follows that $|\m A^+(p)|$ is at most the number of ways to choose up to $(H+1)Z$ crossings of $\pi^+(p)$, which is $O(n)^{(H+1)Z}$.

We can compute $\m A^+(p)$ as follows. First, guess the at most $Z$ crossings for each normal partition node $p'\in\fpath(p)$, and guess one of the at most $((H+1)Z)!$ cyclic orderings of the crossings. Not every cyclic ordering of crossings may correspond to a valid cycle, but it is easy to check in polynomial time whether a cycle exists, and if so, find such a cycle and subsequently compute $\inside(C) \cap \pt^+(p)$.
\end{proof}

We now extend these definitions to pairs of nodes: for 
partition nodes $\{p,p'\}$ in $\calT$, define
\begin{OneLiners}
\item[(i)] $\pt^+(\{p,p'\}):=\pt^+(p) \cup \pt^+(p')$, and
\item[(ii)]
  $\m A^+(\{p,p'\}) := \{ S_p \cup S_{'p} \mid S_p \in\m A^+(p), S_{p'} \in\m A^+(p')\}$; %
  note that
  $\m A^+(\{p,p'\})\supseteq \m A^+(p) \cup \m A^+(p')$ and strict
  containment is possible.
\end{OneLiners}
\Cref{lem:size-A} implies that $|\m A^+(\{p,p'\})| \le n^{O(HZ)}$.
Finally, for
any two nodes $a$ and $b$ in $\calT$, define $\textsf{lca}(a,b)$ as the \emph{lowest
  common ancestor} node of $a$ and $b$ in $\calT$.

\begin{figure} \centering
\begin{subfigure}[t]{.3\textwidth} \centering
  \includegraphics[width=1.5in,page=1]{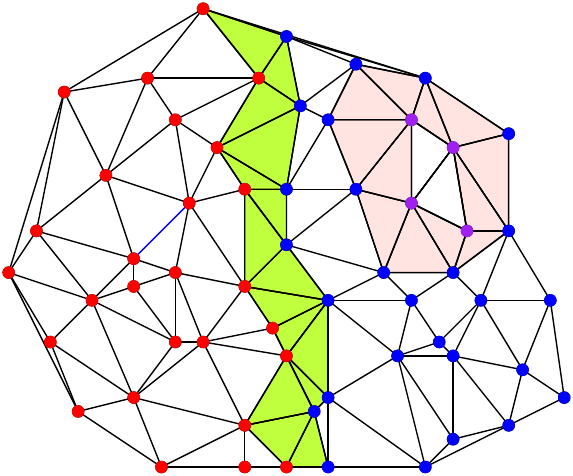}
\end{subfigure}
\quad
\begin{subfigure}[t]{.3\textwidth} \centering
  ~~~\includegraphics[width=1.5in,page=2]{figures/boundary-etc.pdf}
\end{subfigure}
\quad
\begin{subfigure}[t]{.3\textwidth} \centering
  ~~~\includegraphics[width=1.5in,page=3]{figures/boundary-etc.pdf}
\end{subfigure}
\caption{(a) A partition $\pi$ of vertices $U\s V^*$ into three parts (with nodes in
  red, blue, and purple), and with the boundary $\partial(\pi)$ colored
  in. (b) A cycle $C$, and (c) the collection of faces
  $\inside(C) \cap \partial(\pi)$ in bright blue. Since the cycle has
  $8$ edges cut, it is $Z$-amenable and this collection of faces 
  belongs to $\calA(\pi)$ as long as $Z \geq 8$. }
\label{fig:Ap}
\end{figure}

\subsection{Variables and Constraints}
\label{sec:variables}

In this section, we introduce the variables and constraints of our
linear program. If these variables were to take on Boolean values,
they would encode a $Z$-\amenable solution $C$. Of course, the optimal
LP solution will be fractional, and we give the rounding in the next
section. The consistency constraints will be precisely those needed
for the rounding, and should be thought of as defining a suitable
pseudo-distribution that can then be rounded.

Assume we start with a $Z$-\friendly $\m T$.
We begin with defining the $x$-variables. For each partition node $p$ and set of (dual)
faces $S\in\m A^+(p)$, declare a variable $x(\{p\},S)$ with the
constraint
\begin{gather}
  0\le x(\{p\}, S)\le1, \label{eq:x1}
\end{gather}
which represents whether or not the sparsest cut solution $C$ (treated
as a set of faces) satisfies $C\cap \pt^+(p)=S$. In other words,
$x(\{p\},S) = 1$ says that $S$ is the set of faces that lie inside the
target solution, and also belong to $\pt(q)$ for some partition node
$q$ that is either $p$ or an ancestor of $p$.

We also define variables that represent a ``level-two lift'' of these
variables, which capture the same idea for pairs of faces. For each pair of distinct partition nodes $p,p'$ whose lowest
common ancestor $\textsf{lca}(p,p')$ is a partition node, define a variable
$x(\{p,p''\},S)$ with the constraint
\begin{gather}
  0\le x(\{p,p''\},S)\le1, \label{eq:x2}
\end{gather}
which represents whether or not
$C\cap (\pt^+(\{p,p''\}))=S$.

\textbf{Consistency.}
Next we impose ``consistency''
constraints on these variables. For the root node
$\widehat p$, %
we add the constraint %
\begin{gather}
  x(\{\widehat p\},\emptyset) %
  = 1 . \label{eq:consis1}
\end{gather}
 
Recall that cluster partition nodes and cluster nodes alternate in $\m T$, and recall the notion of forcing and relevance from \Cref{def:forcing}. We impose a constraint capturing (i) a ``relaxation'' of
forcings and relevance, where each cluster node fractionally chooses one of its
children partition nodes $p_i$, and (ii) a ``relaxation" of determining which faces in $\pt^+(p_i)$ (for the chosen $p_i$) are contained in $C$. Formally, for each cluster node $c$ whose parent is partition node $p$ and
whose children are the partition nodes $p_1,\lds,p_r$, add the constraint
\BG
x(\{p\},S) = \sum_{\substack{i\in[r], ~S'\in\m A^+(p_i):\\S'\cap \pt^+(p)=S}} x(\{p_i\},S') . \eqnl{assign}
\EG

\textbf{Projections.}  Next, we project the $x(\{p\}, S)$ variables
onto new auxiliary variables to capture whether the sparsest cut
solution contains a dual face $s$, as well as to capture its
intersection with $\m A^+(p)$ on relevant partition nodes $p$. For a
partition node $p$ and a face $s$, define $\m S(p,s)$ as all partition
nodes $p_s$ that (i) are either $p$ itself or descendants of $p$, and (ii) satisfy $s\in \pt( p_s)$.
(If $s$ is a face in $G^*[K(c)]$ where $c$ is the parent cluster node of $p$,
  then $s \in \pt(p_s)$ for exactly one node along each
  path from $p$ to a leaf. On the other hand, if $s$ does not belong
  to $G^*[K(c)]$, then 
$\m S(p,s)=\emptyset$, in which case the
following definition is vacuous and the corresponding $z$ variable can
be removed.)

For each face $s$, for each choice of $D$ being either $\emptyset$ or
$\{s\}$, for each partition node $p$, and for each set of faces
$W\in\m A^+(p)$, define a variable $z(p,s,D,W)$ that captures whether
(a)~$p$ is relevant, (b)~whether or not the sparsest cut solution $C$
contains face $s$ (depending on whether $D=\emptyset$ or $D=\{s\}$),
and (c)~whether $C$ has intersection $W$ with $\pt^+(p)$. We then add
the constraint
\[
  z(p,s,D,W) = \sum_{\substack{p_s\in\m S(p,s), ~S\in\m
      \pt^+(p_s):\\S\cap\{s\}=D\\S\cap \pt^+(p)=W}}x(\{p_s\},S) .
\]

We do an analogous operation of projecting the $x(\{p_1, p_2\}, S)$
variables onto two faces, to capture costs and demands. For each partition
node $p$ and for every two distinct faces $s,t$, consider all pairs of
partition nodes $p_s,p_t$ (not necessarily distinct) satisfying $s\in
\pt( p_s)$ and $t\in \pt( p_t)$ and $\textsf{lca}(p_s,p_t)=p$. Let $\m
S(p,\{s,t\})$ be the collection of sets $\{p_s,p_t\}$ over all such
$p_s,p_t$. (It is possible that $\m S(p,\{s,t\})=\emptyset$, in which
case again the following definition is vacuous and the corresponding variable can be removed.)

For each partition node $p$, for each subset $D\s\{s,t\}$, for each
set $\{p_s,p_t\}\in\m S(p,\{s,t\})$, and for each set $W\in\m A^+(p)$,
define a variable $y(p,\{s,t\},D,W)$ which captures whether (a)~$p$ is
relevant, (b) whether the sparsest cut solution $C$ has intersection
$D$ with $\{s,t\}$ and (c) whether it has intersection $W$ with
$\pt^+(p)$. Add the constraint \BG y(p,\{s,t\},D,W) =
\sum_{\substack{\{p_s,p_t\}\in\m S(p,\{s,t\}), ~S\in\m
    A^+(\{p_s,p_t\}):\\ S\cap\{s,t\}=D\\ S\cap
    \pt^+(p)=W}}x(\{p_s,p_t\},S) . \eqnl{y} \EG We next enforce that
the variables $z(\cd)$ and $y(\cd)$ must have ``consistent marginals''
when viewed as distributions. For each partition node $p$, for each
subset $D\s\{s\}$, and for each set $W\in\m A^+(p)$, impose the
constraint \BG z(p,s,D,W) = \sum_{\substack{D'\s
    \{s,t\}:\\D'\cap\{s\}=D}} y(p,\{s,t\},D',W) .\eqnl{marginals} \EG

\textbf{Marginals.}
Finally, we define variables that do some further projection. We define
a variable $y(\{s,t\})$ capturing the overall event that $s$ and $t$
are separated, and add the consistency constraint
\BG
y(\{s,t\}) = \sum_{\substack{\text{partition node }p,~D\s\{s,t\}:\\|D\cap\{s,t\}|=1,\\W\in\m A^+(p)}}y(p,\{s,t\},D,W) . \eqnl{yst}
\EG
Observe that combining \eqn{y} and \eqn{yst} gives the equality
\BALN
y(\{s,t\}) &= \sum_{\substack{\text{partition node }p,\\D\s\{s,t\}:\\|D\cap\{s,t\}|=1,\\W\in\m A^+(p)}}y(p,\{s,t\},D,W) 
\nonumber\\ &=  \sum_{\substack{\text{partition node }p,\\D\s\{s,t\}:\\|D\cap\{s,t\}|=1,\\W\in\m A^+(p)}} \sum_{\substack{\{p_s,p_t\}\in\m S(p,\{s,t\}),\\S\in\m A^+(\{p_s,p_t\}):\\ S\cap\{s,t\}=D\\ S\cap \pt^+(p)=W}}x(\{p_s,p_t\},S)
\nonumber\\&= \sum_{\substack{\text{partition node }p,\\\{p_s,p_t\}\in\m S(p,\{s,t\}),\\S\in\m A^+(\{p_s,p_t\}):\\|S\cap\{s,t\}|=1}} x(\{p_s,p_t\},S) . \eqnl{yst2}
\EALN

\textbf{Consistency II.}
Finally, for the lifted variables $x(\{p_s,p_t\},S)$, we define an
additional consistency constraint that relates the original variables
of \eqn{x1} to the lifted variables of \eqn{x2}. For each partition
node $p$ and set $W\in\m A^+(p)$ and every pair $\{s,t\}$ of faces, add the constraint
\BG
 x(\{p\},W) = \sum_{\substack{\{p_s,p_t\}\in\m S(p,\{s,t\}):\\S\in\m A^+(\{p_s,p_t\})\\S\cap \pt^+(p)=W}} x(\{p_s,p_t\},S) \eqnl{xfW} .
\EG

\textbf{The Cut Demand, and the Objective.}
We assume that we have an estimate $\alpha$ for the cut demand, which
allows us to impose the constraint:
\BG \sum_{s\ne t}y(\{s,t\})\ge\al . \eqnl{alpha} \EG
Since the edges in the primal corresponding to pairs of faces in the
dual that share a dual edge, the objective function is:
\[ \text{minimize }\sum_{\{s,t\}: s \cap t \in E^*} \cost(\{s,t\}) y(\{s,t\}) ,\]
This completes the definition
of the linear program for the guess $\alpha$.

\BL\leml{feasible} For $\e \leq 1$, let $\m T$ be a $(Z,\e)$-\friendly
nondeterministic hierarchical decomposition.  We can write
$O(\eps^{-1} \log n)$ LPs of the form above, such that one of them has
a feasible solution with fractional cost at most
$(1 + O(\eps))\alpha\, \Phi^*$.  \EL

\begin{proof}
  We write an LP for each setting of $\alpha$ being a power of
  $(1+\eps)$, lying between $1$ and $n^5$, which is an upper bound on
  the separated demand by \Cref{lem:aspectratio}.  By the definition
  of $\m T$ being $(Z,\e)$-\friendly, there exists a forcing $\varphi$
  of the cluster nodes of $\m T$ and a sparsest cut solution $C$ that
  has sparsity at most $(1+\e)$ times the optimal sparsity, and is
  $Z$-\amenable for the forced tree $\calT|_\varphi$.
  Focus on the linear program for $\alpha$ being the largest power of
  $(1+\eps)$ which is no larger than the demand separated by $C$.  Setting the
  variables above according to this solution $C$ and the forcing
  $\varphi$ gives a $0$-$1$ solution each of the linear programs. Now
  the cost of the cycle $C$ is the sparsity of $C$ times the demand
  separated, i.e., at most $(1+\e)\Phi^*  \cdot (1+\eps)\alpha$ as claimed.
\end{proof}

\subsection{Rounding the LP}
\label{sec:rounding-lp}

We round the LP solution top-down, beginning with the root partition node. Our goal is to select a forcing of cluster nodes to their child partition nodes, as well as select an element of $\m A(p)$ for each partition node $p$ that we round, which turns out to be all relevant face nodes under the chosen forcing. Our final solution will be a set $U$ of faces. Initially, $U=\emptyset$; we will add to $U$ at each partition node that we round.

\begin{itemize}
\item For a cluster node $c$ with parent node $p$ and children
  $p_1,\lds,p_r$, we have already rounded $\m A^+(p)$ by this point;
  that is, we have already determined $U\cap \pt(p)$, which we call $W$.
  We assign partition node $p_i$ to cluster node $c$ with probability
  \[ \f1{x(\{p\},W)} \sum_{\substack{S\in\m A^+(p_i):\\ S\cap
        \pt^+(p)=W}} x(\{p_i\},S) .\] This is a probability distribution
  by \eqn{assign}.

\item 
  For the chosen partition node $p_i$, we need to determine which faces in $\pt(p_i)$ are
  in $U$. We want to choose a set $S\in\m A^+(p_i)$
  satisfying $S\cap \pt^+(p)=W$, and then add $S\sm W$ to $U$. We
  simply choose one with probability proportional to
  $x(\{p_i\},S)$. That is, each set $S\in\m A^+(p_i)$ satisfying
  $S\cap \pt^+(p)=W$ is chosen with probability
  \[ x(\{p_i\},S)\cd \bigg( \sum_{\substack{S'\in\m A^+(p_i):\\S'\cap
        \pt^+(p)=W}} x(\{p_i\},S') \bigg) \inv ,\] which is clearly a probability distribution. We then add $S\sm W$ to $U$.
\end{itemize}

We begin with a lemma similar to \cite[Lemma~2.2]{GTW13-spcut}:
\BL\leml{xfS}
For any partition node $p$ and set $S\in\m A^+(p)$, the probability that $p$ is relevant and $U\cap \pt^+(p)=S$ is $x(\{p\},S)$. 
\EL
\BP
We show this by induction from the root down the tree. Let $p'$ be a partition node, let $c$ be a child of $p'$, and let $p$ be a child of $c$. By induction, for any set $W\in\m A^+(p')$, the probability that $p'$ is relevant and $U\cap \pt^+(p')=W$ is $x(\{p'\},W)$.

Fix a set $S\in\m A^+(p)$, define $W:=S\cap \pt^+(p)$, and condition on the event that $U\cap \pt^+(p')=W$, which happens with probability $x(\{p'\},W)$.
Conditioned on this, partition node $p$ is relevant to cluster node $c$ with probability
\[ \f1{x(\{p'\},W)} \sum_{S'} x(\{p\},S') ,\]
where the summation is over all $S'\in\m A^+(p)$ satisfying $S'\cap \pt^+(p')=W$. Conditioned on this as well, set $S\in\m A^+(p)$ is chosen with probability
\[ \f{x(\{p\},S)}{\sum_{S'}x(\{p\},S')} ,\]
with the same summation over $S'$. Unraveling the conditioning, the overall probability of choosing $S\in\m A^+(p)$ is
\[ x(\{p'\},W) \cd \f1{x(\{p'\},W)} \sum_{S'} x(\{p\},S') \cd \f{x(\{p\},S)}{\sum_{S'}x(\{p\},S')} = x(\{p\},S) ,\]
as desired.
\EP

\BC\corl{cost}
For any edge $\{s,t\}$ of the primal graph $G$, the probability that $|U\cap\{s,t\}|=1$ is exactly $y(\{s.t\})$. Therefore, the expected total cost of the primal edges cut is equal to the fractional value of the LP.
\EC
\BP
Consider the partition node $p$ that separates $s$ and $t$;
such a node exists because all leaf cluster nodes are singletons. Since $\{s,t\}$ is an edge of the primal graph, one or both of the dual faces $s$ and $t$ is in $\pt (p)$. Assume without loss of generality that $s \in \pt(p)$. Consider the partition node $p_t$ with $t\in\pt(p_t)$. Since $p_t$ is a descendant of $p$, we have $\m A^+(\{p,p_t\}) = \m A^+(p_t)$, and moreover, both $s$ and $t$ are in $\m A^+(p_t)$. For any $S\in\m A^+(p_t)$ such that $|S\cap\{s,t\}|=1$, by \lem{xfS}, the probability that $p$ is relevant and $U\cap\pt^+(p_t)=S$ is $x(\{p_t\},S)$. These probability events are all disjoint, so the total probability that $|U\cap\{s,t\}|=1$ is
\[ \sum_{\substack{p \text{ separates $s,t$},\\S\in\m A^+(\{p,p_t\}):\\|S\cap\{s,t\}|=1}} x(\{p,p_t\},S) \stackrel{\eqn{yst2}}= y(\{s,t\}). \]
The expected total cost follows by linearity of expectation.
\EP

\BL\leml{demand}
For each $\{s,t\}\in E^*$, we have $\Pr[|U\cap\{s,t\}|=1] =
y(\{s,t\})$, and for each $\{s,t\}\in\m D$, we have
$\Pr[|U\cap\{s,t\}|=1] \ge\f12 \, y(\{s,t\})$.
\EL
\BP
Let $p_s$ be the unique relevant partition node satisfying $s\in \pt( p_s)$, and let $p_t$ be the unique relevant partition node satisfying $t\in \pt( p_t)$. Let $p$ be the lowest common ancestor of $p_s$ and $p_t$, which must be a relevant partition node. Under the randomness of selecting the forcing in the LP rounding, $p_s,p_t,p$ are random variables. %
Define $W:= U\cap \pt( S)$, which is also a random variable. We claim that 
\BG
\E_{p,W} \bigg[ \f1{x(\{p\},W)} \sum_{\substack{\{p_s,p_t\}\in\m S(p,\{s,t\}):\\S\in\m A^+(\{p_s,p_t\}): \\ S\cap \pt^+(p)=W \\ |S\cap\{s,t\}|=1}}x(\{p_s,p_t\},S) \bigg]  = y(\{s,t\}) . \eqnl{Eft}
\EG
This is because the probability of choosing a particular $p,W$ is $x(\{p\},W)$ by \lem{xfS}, and
\BAL
\E_{p,W} &\bigg[ \f1{x(\{p\},W)} \sum_{\substack{\{p_s,p_t\}\in\m S(p,\{s,t\}):\\S\in\m A^+(\{p_s,p_t\}): \\ S\cap \pt^+(p)=W \\ |S\cap\{s,t\}|=1}}x(\{p_s,p_t\},S) \bigg]
\\&= \sum_{p,W}x(\{p\},W) \cd \bigg[ \f1{x(\{p\},W)} \sum_{\substack{\{p_s,p_t\}\in\m S(p,\{s,t\}):\\S\in\m A^+(\{p_s,p_t\}): \\ S\cap \pt^+(p)=W \\ |S\cap\{s,t\}|=1}}x(\{p_s,p_t\},S) \bigg]
\\&= \sum_{p,W} \sum_{\substack{\{p_s,p_t\}\in\m S(p,\{s,t\}):\\S\in\m A^+(\{p_s,p_t\}): \\ S\cap \pt^+(p)=W \\ |S\cap\{s,t\}|=1}}x(\{p_s,p_t\},S)
\\&= \sum_{\substack{\text{partition node }p,\\\{p_s,p_t\}\in\m S(p,\{s,t\}) }} \sum_{\substack{S\in\m A^+(\{p_s,p_t\}): \\ |S\cap\{s,t\}|=1}}x(\{p_s,p_t\},S)
\\&\stackrel{\eqn{yst2}}= y(\{s,t\}) .
\EAL

First, suppose that $\{s,t\}\in E^*$. Then, we must have either
$p_s=p$ or $p_t=p$, since the first partition node ``splitting'' off
$s$ and $t$ must contain either $s$ or $t$.
Let $p'\in\{p_s,p_t\}$ be the node that is not $p$. Then, by \lem{xfS}, for any $p'$ and $S\in\m A^+(p')$, the probability that $p'$ is relevant and $U\cap \pt^+(p)=S$ is $x(\{p'\},S)$. Conditioned on the choices of random variables $p$ and $W$, for any $p'$ that is a descendant of $p$, the probability that $p'$ is relevant and $U\cap \pt^+(p)=S$ is $\f{x(\{p'\},S)}{x(\{p\},W)}$.  Using that $x(\{p'\},S)=x(\{p,p''\},S)$ since $p'$ is a descendant of $p$, this probability is also $\f{x(\{p,p''\},S)}{x(\{p\},W)}$. Also, since $p'$ is a descendant of $p$, we have $\m A^+(p')=\m A^+(\{p,p''\})$. Therefore, conditioned on the choices of $p$ and $W$, the probability that $U$ separates $s$ and $t$ is
\[ \sum_{\substack{p'\text{ descendant of }p,\\S\in\m A^+(p'):\\|S\cap\{s,t\}|=1}} \f{x(\{p,p''\},S)}{x(\{p\},W)} = \sum_{\substack{\{p_s,p_t\}\in\m S(p,\{s,t\}),\\S\in\m A^+(\{p_s,p_t\}): \\ S\cap \pt^+(p)=W \\ |S\cap\{s,t\}|=1}}\f{x(\{p_s,p_t\},S)}{x(\{p\},W)},\]
which is exactly the term inside the expectation in \eqn{Eft}. Unraveling the conditioning on the choices of $p$ and $W$ and using \eqn{Eft}, the probability that $U$ separates $s$ and $t$ is $y(\{s,t\})$, as desired.

Now consider a pair $\{s,t\}\in\m D$. Fix partition nodes $p_s\in\m S(p,s)$ with $s\in \pt( s)$ and $t\in \pt( t)$. Conditioned on the choices of $p$ and $W\in\m A^+(p)$, for each $p_s\in\m S(p,s)$ and $S\in\m A^+(p_s)$ satisfying $S\cap \pt^+(p)=W$, the probability that both $p_s$ is relevant and $U\cap\m A^+(p_s)=S$ is $\f{x(\{p_s\},S)}{x(\{p\},W)}$ by \lem{xfS}. Let $Z_s$ denote the random variable $U\cap\{s\}$; then, for each $D_s\s\{s\}$,%
$\Pr[Z_s=D_s \mid p,W]$ \[ =  \f1{x(\{p\},W)} \sum_{\substack{p_s\in\m S(p,s),\\S\in\m A^+(p_s):\\S\cap \pt^+(p)=W\\S\cap\{s\}=D_s}}x(\{p_s\},S) = \f{z(p,s,D_s,W)}{x(\{p\},W)} .\]
Similarly, conditioned on the choices of $p$ and $W$, for each $p_t\in\m S(p,t)$ and $T\in\m A^+(p_t)$ satisfying $T\cap \pt^+(p)=W$, the probability that both $p_t$ is relevant and $U\cap\m A^+(p_t)=T$ is $\f{x(\{p_t\},T)}{x(\{p\},W)}$. Let $Z_t$ denote the random variable $U\cap\{t\}$; then, for each $D_t\s\{t\}$,%
$\Pr[Z_t=D_t \mid p,W]$ \[ = \f1{x(\{p\},W)} \sum_{\substack{p_t\in\m S(p,t),\\T\in\m A^+(p_t):\\T\cap \pt^+(p)=W\\T\cap\{t\}=D_t}}x(\{p_t\},T) = \f{z(p,t,D_t,W)}{x(\{p\},W)} .\]
Observe that for a given $\{p_s,p_t\}\in\, S(p,\{s,t\})$ and $S\in\m A^+(p_s)$ and $T\in\m A^+(p_t)$ satisfying $S\cap \pt^+(p)=T\cap \pt^+(p)=W$, the sets $S\sm W$ and $T\sm W$ are disjoint. By the nature of the LP rounding algorithm, we have that conditioned on the choices of $p$ and $W$, the random variables $Z_s$ and $Z_t$ are independent.

Next, conditioned on the choices of $p$ and $W$, let $Y$ be a random variable that takes each value $D\s\{s,t\}$ with probability $\f{y(p,\{s,t\},D,W)}{x(\{p\},W)}$. This is a probability distribution because
\BAL
&\sum_{D\s\{s,t\}}\Pr[Y=D\mid p,W]
= \sum_{D\s\{s,t\}}\f{y(p,\{s,t\},D,W)}{x(\{p\},W)}
\\&=
\sum_{D\s\{s,t\}} \f1{x(\{p\},W)} \sum_{\substack{\{p_s,p_t\}\in\m S(p,\{s,t\}),\\S\in\m A^+(\{p_s,p_t\}):\\S\cap\{s,t\}=D\\S\cap \pt^+(p)=W}} x(\{p_s,p_t\},S) 
\\&= \f1{x(\{p\},W)} \sum_{\substack{\{p_s,p_t\}\in\m S(p,\{s,t\}),\\S\in\m A^+(\{p_s,p_t\}):\\S\cap \pt^+(p)=W}} x(\{p_s,p_t\},S)
\\&\stackrel{\eqn{xfW}}= 1.
\EAL
We now claim that $Y$, now viewed as a distribution on $\{\emptyset,\{s\}\}\times\{\emptyset,\{t\}\}$ (instead of on the power set of $\{s,t\}$), agrees with $Z_s(\cd)$ and $Z_t(\cd)$ on the marginals. Namely, for $D_s\s\{s\}$, we have
\BALN
\sum_{\substack{D\s\{s,t\}:\\D\cap\{s\}=D_s}}\Pr[Y=D\mid p,W] &= \sum_{\substack{D\s\{s,t\}:\\D\cap\{s\}=D_s}} \f{y(p,\{s,t\},D,W)}{x(\{p\},W)}\nonumber
\\&\stackrel{\eqn{marginals}}= \f{z(p,s,D,W)}{x(\{p\},W)}\nonumber
\\&=\Pr[Z_s=D_s\mid p,W] ,\eqnl{margs}
\EALN
and for $D_t\s\{t\}$, we similarly have
\BG
\sum_{\substack{D\s\{s,t\}:\\D\cap\{t\}=D_s}}\Pr[Y=D\mid p,W] =\Pr[Z_t=D_t\mid p,W] .\eqnl{margt}
\EG

\begin{claim}
  \label{clm:half}
  Define event $\calE$ to be $((Z_s=\{s\})\wedge (Z_t=\emptyset))
  \vee ((Z_s=\emptyset) \wedge (Z_t=\{t\}))$. Then
  \begin{gather}
    \Pr[\;\calE \mid p,W] \ge \f12\,\Pr[(Y=\{s\}) \vee (Y=\{t\})\mid
    p,W] ,\eqnl{half}
  \end{gather}
\end{claim}

Assuming \Cref{clm:half}, the probability that solution $U$ separates $s$ and $t$ conditioned on $p$ and $W$ is
\BAL
\Pr[ &|U\cap\{s,t\}|=1\mid p,W] \\&=\Pr[\; \calE \mid p,W] 
\\&\ge\f12\,\Pr[(Y=\{s\}) \vee (Y=\{t\})\mid p,W] 
\\&= \f12\,\lp\f{y(p,\{s,t\},\{s\},W)}{x(\{p\},W)}+\f{y(p,\{s,t\},\{t\},W)}{x(\{p\},W)}\rp
\\&= \f12\,\bigg(\f1{x(\{p\},W)} \sum_{\substack{\{p_s,p_t\}\in\m S(p,\{s,t\}),\\S\in\m A^+(\{p_s,p_t\}): \\ S\cap \pt^+(p)=W \\ |S\cap\{s,t\}|=1}}x(\{p_s,p_t\},S),\bigg) 
\EAL
which is half of the term inside the expectation in \eqn{Eft}. Finally, unraveling the conditioning on the choices of $p$ and $W$ and using \eqn{Eft}, we obtain
\[ \Pr[|U\cap\{s,t\}|=1] \ge\f12\,y(\{s,t\}) ,\]
which concludes \lem{demand}.
\EP

\BC
Let \textsf{cost} be the total cost of the primal edges cut, and let \textsf{demand} be the total demands cut. Let $\Phi^*$ be the optimal sparsity. Then,
$\E[\mathsf{cost} - 2 (1+O(\e)) \cd \Phi^* \cd \mathsf{demand} ] \le 0$.
\EC
\BP
By \thm{structure}, there is a $t$-nondeterministic $(Z,\e)$-friendly hierarchical decomposition $\m T$. Given $\m T$, \lem{feasible} says that there is a feasible solution to the LP with fractional demand at least $\al$ and fractional cost at most $(1+O(\e))\al \Phi^*$. By \cor{cost}, $\E[\mathsf{cost}]$ equals the fractional cost of the LP, which is at most $ (1+O(\e))\al\Phi^* $, and by \lem{demand}, $\E[\mathsf{demand}] \ge \f12\,\al $, so we are done. 
\EP

This implies that
$\E[\mathsf{cost}]/\E[\mathsf{demand}] \leq 2(1+O(\eps))\cdot
\Phi^*$.
While this ratio of expectations is normally not enough, we can
use~\Cref{lem:roe-to-eor} to find a cut of sparsity at most
$2(1+O(\e))\Phi^*$ with high probability, and hence complete the algorithm.
It just remains to prove \Cref{clm:half}.
\begin{proof}[Proof of \Cref{clm:half}]
  Define
\begin{alignat}{2}
a &:= \Pr[Y=\{\emptyset\}\mid p,W]  & \qquad 
b &:= \Pr[Y=\{s\}\mid p,W] \\
c &:= \Pr[Y=\{t\}\mid p,W] & 
d &:= \Pr[Y=\{s,t\}\mid p,W],
\end{alignat}
so that $a+b+c+d=1$.
By \eqn{margs}~and~\eqn{margt},
\BAL
\Pr[Z_s=\emptyset\mid p,W] = \sum_{\substack{D\s\{s,t\}:\\D\cap\{s\}=\emptyset}}\Pr[Y=D\mid p,W] = a+c , \\ 
\Pr[Z_s=\{s\}\mid p,W] = \sum_{\substack{D\s\{s,t\}:\\D\cap\{s\}=\{s\}}}\Pr[Y=D\mid p,W] = b+d , \\
\Pr[Z_t=\emptyset\mid p,W] = \sum_{\substack{D\s\{s,t\}:\\D\cap\{t\}=\emptyset}}\Pr[Y=D\mid p,W] = a+b , \\
\Pr[Z_t=\{t\}\mid p,W] = \sum_{\substack{D\s\{s,t\}:\\D\cap\{t\}=\{t\}}}\Pr[Y=D\mid p,W] = c+d .
\EAL
Since $Z_s$ and $Z_t$ are independent conditioned on the choices of $p$ and $W$, 
\BAL
&\Pr[((Z_s=\{s\})\wedge (Z_t=\emptyset)) \vee ((Z_s=\emptyset) \wedge (Z_t=\{t\}))\mid p,W] \\={} & \Pr[Z_s=\{s\}\mid p,W]  \cd \Pr[Z_t=\emptyset\mid p,W]
\\&~~~~~~~~~+ \Pr[Z_s=\emptyset\mid p,W] \cd \Pr[Z_t=\{t\}\mid p,W]
\\ ={} & (b+d)(a+b) + (a+c)(c+d)
\EAL
Moreover,
\[ \Pr[(Y=\{s\}) \vee (Y=\{t\})\mid p,W] = b+c ,\]
so \eqn{half} now follows from
\Cref{prop:decoupling}.
This concludes the proof of \lem{demand}.
\end{proof}

}

\section{Concluding Remarks}

A natural question is whether there is a polynomial-time $2+\epsilon$
approximation algorithm. Another is whether a 2-approximation or
even better can be achieved in quasi-polynomial time.
Note that no approximation ratio smaller than two is known even for the special case of
series-parallel graphs (which are planar and have treewidth two). The
greatest approximation lower bound known (assuming the unique games conjecture) is
$\approx (1.139 - \eps)$, via a relatively simple reduction is from
\textsc{Max-Cut}~\cite{GTW13-spcut}. Given the known limitations of linear programming
techniques for \textsc{Max-Cut}, we may need to use semidefinite
programs to obtain an approximation ratio better than two.
Another question is whether the result can be extended to more general families
of graphs such as minor-free graphs. %

{\small
\bibliographystyle{alpha}
\bibliography{bib}

\newcommand{\etalchar}[1]{$^{#1}$}
\begin{thebibliography}{AGMW10}

\bibitem[ACAK20]{AbboudSparsest}
Amir Abboud, Vincent Cohen-Addad, and Philip~N. Klein.
\newblock New hardness results for planar graph problems in {P} and an
  algorithm for sparsest cut.
\newblock In {\em STOC '20}, 2020.

\bibitem[AGG{\etalchar{+}}19]{AGGNT14-sicomp}
Ittai Abraham, Cyril Gavoille, Anupam Gupta, Ofer Neiman, and Kunal Talwar.
\newblock Cops, robbers, and threatening skeletons: Padded decomposition for
  minor-free graphs.
\newblock {\em {SIAM} J. Comput.}, 48(3):1120--1145, Jun 2019.

\bibitem[AGMW10]{AbrahamGMW10}
Ittai Abraham, Cyril Gavoille, Dahlia Malkhi, and Udi Wieder.
\newblock Strong-diameter decompositions of minor free graphs.
\newblock {\em Theory Comput. Syst.}, 47(4):837--855, 2010.

\bibitem[ALN08]{ALN05}
Sanjeev Arora, James~R. Lee, and Assaf Naor.
\newblock Euclidean distortion and the sparsest cut.
\newblock {\em J. Amer. Math. Soc.}, 21(1):1--21, 2008.

\bibitem[AR98]{AR98}
Yonatan Aumann and Yuval Rabani.
\newblock An {$O(\log k)$} approximate min-cut max-flow theorem and
  approximation algorithm.
\newblock {\em SIAM J. Comput.}, 27(1):291--301, 1998.

\bibitem[Aro97]{Arora97}
Sanjeev Arora.
\newblock Nearly linear time approximation schemes for euclidean {TSP} and
  other geometric problems.
\newblock In {\em 38th Annual Symposium on Foundations of Computer Science,
  {FOCS} '97, Miami Beach, Florida, USA, October 19-22, 1997}, pages 554--563.
  {IEEE} Computer Society, 1997.

\bibitem[ARV09]{ARV04}
Sanjeev Arora, Satish Rao, and Umesh Vazirani.
\newblock Expander flows, geometric embeddings and graph partitioning.
\newblock {\em J. ACM}, 56(2):Art. 5, 37, 2009.

\bibitem[BGK16]{BartalGK16}
Yair Bartal, Lee{-}Ad Gottlieb, and Robert Krauthgamer.
\newblock The traveling salesman problem: Low-dimensionality implies a
  polynomial time approximation scheme.
\newblock {\em {SIAM} J. Comput.}, 45(4):1563--1581, 2016.

\bibitem[CFW12]{CFW12}
Amit Chakrabarti, Lisa Fleischer, and Christophe Weibel.
\newblock When the cut condition is enough: a complete characterization for
  multiflow problems in series-parallel networks.
\newblock In {\em STOC}, pages 19--26, 2012.

\bibitem[CGN{\etalchar{+}}06]{CGNRS01}
Chandra Chekuri, Anupam Gupta, Ilan Newman, Yuri Rabinovich, and Alistair
  Sinclair.
\newblock Embedding {$k$}-outerplanar graphs into {$\ell_1$}.
\newblock {\em SIAM J. Discrete Math.}, 20(1):119--136, 2006.

\bibitem[CGR08]{CGR05}
Shuchi Chawla, Anupam Gupta, and Harald R{\"a}cke.
\newblock Embeddings of negative-type metrics and an improved approximation to
  generalized sparsest cut.
\newblock {\em ACM Trans. Algorithms}, 4(2):Art. 22, 18, 2008.

\bibitem[CJLV08]{CJLV08}
Amit Chakrabarti, Alexander Jaffe, James~R. Lee, and Justin Vincent.
\newblock Embeddings of topological graphs: Lossy invariants, linearization,
  and 2-sums.
\newblock In {\em FOCS}, pages 761--770, 2008.

\bibitem[CKK{\etalchar{+}}06]{CKKRS05}
Shuchi Chawla, Robert Krauthgamer, Ravi Kumar, Yuval Rabani, and D.~Sivakumar.
\newblock On the hardness of approximating multicut and sparsest-cut.
\newblock {\em Comput. Complexity}, 15(2):94--114, 2006.

\bibitem[CKR10]{CKR10}
Eden Chlamt\'a\v{c}, Robert Krauthgamer, and Prasad Raghavendra.
\newblock Approximating sparsest cut in graphs of bounded treewidth.
\newblock In {\em APPROX}, volume 6302 of {\em LNCS}, pages 124--137. 2010.

\bibitem[CSW13]{CSW10}
Chandra Chekuri, F.~Bruce Shepherd, and Christophe Weibel.
\newblock Flow-cut gaps for integer and fractional multiflows.
\newblock {\em J. Combin. Theory Ser. B}, 103(2):248--273, 2013.

\bibitem[GNRS04]{GNRS99}
Anupam Gupta, Ilan Newman, Yuri Rabinovich, and Alistair Sinclair.
\newblock Cuts, trees and {$\ell_1$}-embeddings of graphs.
\newblock {\em Combinatorica}, 24(2):233--269, 2004.

\bibitem[GTW13]{GTW13-spcut}
Anupam Gupta, Kunal Talwar, and David Witmer.
\newblock On the non-uniform sparsest cut problem on bounded treewidth graphs.
\newblock In {\em STOC}, pages 281--290, 2013.

\bibitem[H{\aa}s01]{Hastad01}
Johan H{\aa}stad.
\newblock Some optimal inapproximability results.
\newblock {\em J. ACM}, 48(4):798--859, 2001.

\bibitem[KARR90]{KleinARR90}
Philip~N. Klein, Ajit Agrawal, R.~Ravi, and Satish Rao.
\newblock Approximation through multicommodity flow.
\newblock In {\em Proceedings of the 31st Annual Symposium on Foundations of
  Computer Science}, pages 726--737, 1990.

\bibitem[KKMO07]{KKMO}
Subhash Khot, Guy Kindler, Elchanan Mossel, and Ryan O'Donnell.
\newblock Optimal inapproximability results for {MAX}-{CUT} and other
  2-variable {CSP}s?
\newblock {\em SIAM J. Comput.}, 37(1):319--357, 2007.

\bibitem[KM]{planarity}
Philip~N. Klein and Shay Mozes.
\newblock Optimization algorithms for planar graphs -- http://planarity.org/.

\bibitem[KPR93]{KPR93}
Philip Klein, Serge~A. Plotkin, and Satish~B. Rao.
\newblock Excluded minors, network decomposition, and multicommodity flow.
\newblock In {\em STOC}, pages 682--690, 1993.

\bibitem[KRAR95]{KleinRAR95}
Philip~N. Klein, Satish Rao, Ajit Agrawal, and R.~Ravi.
\newblock An approximate max-flow min-cut relation for unidirected
  multicommodity flow, with applications.
\newblock {\em Combinatorica}, 15(2):187--202, 1995.

\bibitem[KV15]{KV05}
Subhash~A. Khot and Nisheeth~K. Vishnoi.
\newblock The unique games conjecture, integrability gap for cut problems and
  embeddability of negative-type metrics into {$\ell_1$}.
\newblock {\em J. ACM}, 62(1):Art. 8, 39, 2015.

\bibitem[LLR95]{LLR95}
Nathan Linial, Eran London, and Yuri Rabinovich.
\newblock The geometry of graphs and some of its algorithmic applications.
\newblock {\em Combinatorica}, 15(2):215--245, 1995.

\bibitem[LN06]{LeeNaor06}
James~R. Lee and Assaf Naor.
\newblock $l_p$ metrics on the {Heisenberg} group and the {Goemans}-{Linial}
  conjecture.
\newblock In {\em FOCS}, pages 99--108, 2006.

\bibitem[LR10]{LR10}
James~R. Lee and Prasad Raghavendra.
\newblock Coarse differentiation and multi-flows in planar graphs.
\newblock {\em Discrete Comput. Geom.}, 43(2):346--362, 2010.

\bibitem[LS09]{LS09}
James~R. Lee and Anastasios Sidiropoulos.
\newblock On the geometry of graphs with a forbidden minor.
\newblock In {\em S{TOC}}, pages 245--254. 2009.

\bibitem[LS13]{LeeS13}
James~R. Lee and Anastasios Sidiropoulos.
\newblock Pathwidth, trees, and random embeddings.
\newblock {\em Combinatorica}, 33(3):349--374, 2013.

\bibitem[MS90]{MS90}
David~W. Matula and Farhad Shahrokhi.
\newblock Sparsest cuts and bottlenecks in graphs.
\newblock {\em Discrete Appl. Math.}, 27(1-2):113--123, 1990.

\bibitem[NY17]{NaorY17}
Assaf Naor and Robert Young.
\newblock The integrality gap of the {Goemans-Linial} {SDP} relaxation for
  sparsest cut is at least a constant multiple of {$\sqrt{\log n}$}.
\newblock In {\em STOC}, pages 564--575, 2017.

\bibitem[NY18]{NaorY18}
Assaf Naor and Robert Young.
\newblock Vertical perimeter versus horizontal perimeter.
\newblock {\em Ann. of Math. (2)}, 188(1):171--279, 2018.

\bibitem[OS81]{OS81}
Haruko Okamura and P.~D. Seymour.
\newblock Multicommodity flows in planar graphs.
\newblock {\em J. Combin. Theory Ser. B}, 31(1):75--81, 1981.

\bibitem[Pat13]{Patel13}
Viresh Patel.
\newblock Determining edge expansion and other connectivity measures of graphs
  of bounded genus.
\newblock {\em {SIAM} J. Comput.}, 42(3):1113--1131, 2013.

\bibitem[PP93]{Park1993FindingMC}
James~K. Park and Cynthia~A. Phillips.
\newblock Finding minimum-quotient cuts in planar graphs.
\newblock In {\em STOC '93}, 1993.

\bibitem[PT95]{PT95}
Serge Plotkin and {\'E}va Tardos.
\newblock Improved bounds on the max-flow min-cut ratio for multicommodity
  flows.
\newblock {\em Combinatorica}, 15(3):425--434, 1995.

\bibitem[Rao92]{Rao92}
Satish Rao.
\newblock Faster algorithms for finding small edge cuts in planar graphs
  (extended abstract).
\newblock In S.~Rao Kosaraju, Mike Fellows, Avi Wigderson, and John~A. Ellis,
  editors, {\em Proceedings of the 24th Annual {ACM} Symposium on Theory of
  Computing, May 4-6, 1992, Victoria, British Columbia, Canada}, pages
  229--240. {ACM}, 1992.

\bibitem[Rao99]{Rao99}
Satish Rao.
\newblock Small distortion and volume preserving embeddings for planar and
  {E}uclidean metrics.
\newblock In {\em SOCG}, pages 300--306, 1999.

\end{thebibliography}

\appendix

\section{Missing Proofs}
\label{sec:missing-proofs}

\subsection{The Decoupling Lemma} \label{sec:decouple}

\begin{proposition}
  \label{prop:decoupling}
  Given non-negative numbers $a,b,c,d$ with sum $a+b+c+d=1$,
  \begin{gather}
    L := (a+b)(b+d) + (a+c)(c+d) \ge \nicefrac12\cdot (b+c) \label{half2}
  \end{gather}
\end{proposition}

\begin{proof}
We have
\begin{align*}
1 = (a+b+c+d)^2 &= a^2+b^2+c^2+d^2+2(ab+ac+ad+bc+bd+cd)
\\ &= 2(ab+ac+2ad + b^2+bd + c^2+cd ) + a^2+d^2-2ad-b^2-c^2+2bc
\\ &= 2L + (a-d)^2 - (b-c)^2,
\end{align*}
which gives now gives the claimed bound:
\begin{gather*}
  b+c = 1-(a+d) = \big(2L + (a-d)^2 - (b-c)^2\big) -(a+d) \le 2L.
\end{gather*}
The inequality above uses that $(b-c)^2\ge0$ and $(a-d)^2\le a+d$ for
$a,d\in[0,1]$. To see the latter, let $a \leq d$ without loss of
generality and let $\eps = d-a$. Since $\eps \in [0,1]$ it follows
that $\eps^2 \leq
\eps \leq 2a+\eps$. 
\end{proof}

\subsection{The Expectation Lemma}

\begin{lemma}
  \label{lem:roe-to-eor}
  Given a randomized algorithm $A$ for sparsest cut which outputs a
  cut $(U, V \setminus U)$ that has
  $\E[\cost(U)]/\E[\demand(U)] \leq \alpha$ for some
  polynomially-bounded $\alpha$, there is a procedure that outputs a
  cut with $\cost(U)/\demand(U) \leq (1+o(1))\alpha $ in polynomial
  time, with probability $1-\nf1n$.
\end{lemma}

\begin{proof}
  There are several ways to perform this conversion: e.g., we can the
  approach of~\cite{CKR10,GTW13-spcut} of derandomizing the algorithm
  to find the desired cut. But here is a different and conceptually
  simpler approach. Firstly, by running $A$ independently several
  times until it returns a cut with finite sparsity, we can assume
  that $A$ always returns a cut that separates some demands. (This is
  just conditioning on the event that $A$ outputs a finite-sparsity
  cut.) Because all edge costs and demands are polynomially bounded
  by~\Cref{lem:aspectratio}, we need to repeat $A$ at most
  polynomially many times for this to happen. Also, the ratio of the
  expected cost to the expected demand does not increase by this
  conditioning.

  Run the algorithm independently $N$ times to get cuts
  $U_1, \ldots, U_N$, and return the sparsest cut among the ones
  returned. By~\Cref{lem:aspectratio} both the costs and demands of
  the cuts are bounded between $1$ and $n^5$; hence if $N = n^c$ for
  suitably large constant $c$, a Chernoff bound implies that the
  average cost and demand among these $N$ samples are both within a
  $(1\pm \nf1n)$ factor of their expectations with high probability,
  and hence
  \[ \frac{\sum_{i} \cost(U_i)}{\sum_{i} \demand(U_i)}
    \leq \frac{\alpha(1+\nf1n)}{(1-\nf1n)}. \]
  Now using that $\min_i \cost(U_i)/\demand(U_i) \leq \frac{\sum_{i}
    \cost(U_i)}{\sum_{i} \demand(U_i)}$ completes the argument.
\end{proof}

\subsection{Proofs from \texorpdfstring{\Cref{sec:preliminaries}}{Section on Preliminaries}}
\label{app:prelims}

\begin{proof}[Proof of Lemma~\ref{lem:aspectratio}]
  Let $G$ be an instance of sparsest cut on a planar graph with $n$
  verticeas. We guess the edge $e_{\max}$ with largest cost that is
  part of the sparsest cut $U^*$ in $G$, and \emph{contract} all edges
  with cost larger than that of $e_{\max}$ (keeping parallel edges,
  and summing up the demand at the merged vertices). Moreover we round
  up the cost of each remaining edge to the closest multiple of
  $\cost(e_{\max})/n^2$ greater than it.  Let $G'$ be the resulting
  instance. Since no edge of the sparsest cut $U^*$ has been
  contracted during the procedure, it remains is a feasible solution
  for $G'$, separating the same demand as in $G$. Moreover, its cost
  in $G'$ is at most the cost of $U^*$ in $G$, plus the increase due
  to rounding up the edge costs. There are at most $O(n)$ such edges
  by planarity, so the total cost increase is at most
  $\cost(e_{\max})/n$.

  It follows that the sparsity of each cut in $G'$ is only higher than
  the corresponding cut in $G$, and moreover, there exists a cut in
  $G'$ of sparsity at most $(1+\nf1n)$ times the optimal sparsity in
  $G$. Therefore, applying $\calA$ on $G'$ yields a solution for $G$
  of sparsity at most $\alpha(1+o(1))$ times the sparsity of the
  sparsest cut. Now running $\calA$ for all $n$ possible choices of
  $e_{\max}$, and outputing the best solution gives
  $\alpha(1+o(1))$-approximation algorithm with running time
  $O(n \cdot T_n)$ as claimed. Moreover, we can divide all edge costs
  in $G'$ down by $\cost(e_{\max})/n^2$ to ensure that the costs are
  integers in the range $\{1, \ldots, n^2\}$ without changing the
  approximation factor for any cut.

  Similarly, we can guess the largest demand $\set{u,v}$ that is cut
  by the sparsest cut, \emph{delete} the demands larger than this
  value, and round \emph{down} all demands to the closest multiple of
  $\demand(\set{u,v})/n^3$. This again means the sparsity of any cut
  with the changed demands is at least that with the original demands,
  but that of the sparsest cut only increases by a factor of
  $1/(1-\nf1n) \approx 1+\nf1n$. Again, guessing this largest demand
  can be done with $O(n^2)$ runs.
\end{proof}

\begin{proof}[Proof of~\Cref{lem:simple-cut}]
  If $G[U]$ is not connected, consider its connected components
  $U_1, \ldots, U_k$. We get $\sum_i |\cost(U_i)| = |\cost(U)|$, but
  $\sum_i \demand(U_i) \geq \demand(U)$, since the demands that cross
  between components of $G[U]$ are counted on the left but not on the
  right. Hence, $\min_i \sparsity(U_i) \leq \sparsity(U)$. Now, if
  $G[U]$ is connected, apply the same argument to its complement.
\end{proof}

\subsection{Proof of Lemmas~\ref{lem:preserve-demand-separation} and~\ref{lem:good-cycle}}
\label{sec:planar}

For a subset of edges $S$ in $G^*$, let
$\chi_S \in (\bbF_2)^{|E(G^*)|}$ denote the characteristic vector of
the set $S$. 

\PreserveDemand*

\begin{proof} Suppose $x$ and $y$ are separated by $\widehat C$.  Let $P$
  be a simple $x$-to-$y$ path in $G$ such that exactly one edge of $P$
  is in $\widehat C$.  For any multiset $S$ of edges, let $\rho(S)$ be
  the cardinality of the intersection of $S$ with $P$.  We have
  $\rho(\widehat C) \equiv 1 \pmod{2}$.

  By the property of $C_1, \ldots, C_k$, we have
  $\sum_{i=1}^k \rho(C_i) \equiv 1 \pmod{2}$.  Therefore there exists
  $i$ such that $\rho(C_i) \equiv 1 \pmod{2}$.  The lemma then follows
  from the fact that if $e$ occurs an odd number of times in $S$ and
  $x$ and $y$ are the endpoints of $e$ in $G$ then $\phi_S(x)\neq
  \phi_S(y)$.
\end{proof}

\ChooseCycle*

\begin{proof}
By assumption,
\begin{equation}
  \sum_{i=1}^k d_G(U(C_i)) \geq d_G(U(\widehat C))
\end{equation}
Assume for a contradiction that $U(C_i)$ has sparsity greater than
$(1+\eps) s$ for $i=1, \ldots, k$.
Then, for $i=1, \ldots, k$,
$$c_G(U(C_i)) > (1+\eps)\,s \cdot d_G(U(C_i))$$
Summing, we obtain
$$\sum_i c_G(U(C_i)) > (1+\eps)\,s \cdot \sum_i d_G(U(C_i))$$
The left-hand side is at most $(1+\eps)\;c_G(U(\widehat C))$, and the
sum on the right-hand side is at least $d_G(U(\widehat C))$, so
$$(1+\eps)\; c_G(U(\widehat C))/d_G(U(\widehat C)) > (1+\eps)\,s$$
which is a contradiction.
\end{proof}

\end{document}